\documentclass[runningheads,orivec,abstracton]{scrartcl}
\usepackage{authblk}
\usepackage[small]{caption}
\usepackage[utf8]{inputenc}
\usepackage{cite}
\usepackage{subfig}
\usepackage{wrapfig}
\usepackage{xcolor}
\usepackage{xspace}
\usepackage{enumerate}
\usepackage{microtype}
\def\wasyfamily{\fontencoding{U}\fontfamily{wasy}\selectfont}
\def\hexagon{\hbox{\wasyfamily\char55}}
\usepackage{amsmath}
\usepackage{amssymb}
\usepackage{amstext}
\usepackage{esvect}
\usepackage{amsthm}
\newtheorem{fact}{Fact}
\newtheorem{observation}{Observation}
\newtheorem{lemma}{Lemma}
\newtheorem{theorem}{Theorem}
\newtheorem{corollary}{Corollary}
\newtheorem{definition}{Definition}
\newtheorem{claim}{Claim}
\usepackage{amsopn}
\usepackage{hyperref}
\usepackage{graphicx}
\usepackage{paralist}
\usepackage{enumitem}
\newenvironment{myenum}
  {\begin{enumerate}[leftmargin=0cm,itemindent=.5cm,labelwidth=\itemindent,labelsep=0cm,align=left]}
  {\end{enumerate}}
\usepackage[textsize=footnotesize,disable]{todonotes}  

\DeclareMathOperator{\front}{front} 

\newcommand{\hp}[2]{\ensuremath{l_{#1 #2}^+}}
\newcommand{\eps}{\ensuremath{\varepsilon}}

\newcommand{\sap}{self-approaching path\xspace}
\newcommand{\saps}{self-approaching paths\xspace}
\newcommand{\sa}{self-ap\-proaching\xspace}

\newcommand{\sad}{self-approaching drawing\xspace}
\newcommand{\sads}{self-approaching drawings\xspace}
\newcommand{\ic}{increasing-chord\xspace}

\newcommand{\R}{\ensuremath{\mathbb R}}
\newcommand{\Hp}{\ensuremath{\mathbb H^2}}

\newcommand{\ifao}{if and only if\xspace}
\newcommand{\vright}{\vec{e_1}}
\newcommand{\vup}{\vec{e_2}}\newcommand{\ray}[2]{\ensuremath{\textnormal{ray}(#1,#2)}}
\newcommand{\bc}{BC\xspace}

\newcommand{\rt}[1]{\ensuremath{r(#1)}}
\newcommand{\dir}[1]{\ensuremath{\textnormal{dir}(#1)}}
\newcommand{\subcact}[1]{\ensuremath{G^{#1}}}
\newcommand{\subcactb}[2]{\ensuremath{G^{#1}_{#2}}}

\newcommand{\dg}{^\circ}

\newcommand{\aschnyder}[1]{#1-Schnyder\xspace}

\newcommand{\wlg}{without loss of generality}
\newcommand{\ang}[2]{\angle(#1,#2)}
\newcommand{\dist}[2]{\textnormal{dist}(#1,#2)}
\newcommand{\pred}[1]{\pi(#1)}
\newcommand{\gpred}[2]{\pi^{#1}(#2)}
\newcommand{\sublem}[1]{\lowercase\expandafter{(\romannumeral #1)\relax}\xspace}
\newcommand{\Wlog}{Without loss of generality}
\newcommand{\concat}[2]{#1.#2}
\renewcommand{\vec}[1]{\vv{#1}}
\newcommand{\vecl}[1]{\vec{#1}}
\newcommand{\upwdir}[1]{U(#1)}
\newcommand{\downwdir}[1]{D(#1)}
\newcommand{\colvect}[2]{(#1,#2)^\top}
\newcommand{\matr}[4]{%
  \ensuremath{\bigl(\negthinspace\begin{smallmatrix}#1&#2\\#3&#4\end{smallmatrix}\bigr)}}
\newcommand{\regred}[1]{R_{#1}^r}
\newcommand{\reggreen}[1]{R_{#1}^g}
\newcommand{\regblue}[1]{R_{#1}^b}
\newcommand{\angleccw}[2]{\angle_\textnormal{ccw}(#1,#2)}

\title{On Self-Approaching and Increasing-Chord Drawings of 3-Connected Planar Graphs\thanks{A preliminary version of this paper appeared at the 22nd International Symposium on Graph Drawing, Würzburg, Germany}
}
 
\author{Martin N\"ollenburg, Roman Prutkin, and Ignaz Rutter}  
\affil{Institute of Theoretical Informatics,\\ Karlsruhe Institute of Technology, Germany}
\date{}
\begin{document}

\maketitle

\begin{abstract}
    An $st$-path in a drawing of a graph is \emph{self-approaching} if during the traversal
  of the corresponding curve from $s$ to any point $t'$ on the curve the
  distance to $t'$ is non-increasing.  A path has \emph{increasing chords} if it
  is self-approaching in both directions.  A drawing is self-approaching
  (increasing-chord) if any pair of vertices is connected by a
  self-approaching (increasing-chord) path.

  We study self-approaching and increasing-chord drawings of
  triangulations and 3-connected planar graphs.  We show that in the
  Euclidean plane, triangulations admit increasing-chord drawings, and
  for planar 3-trees we can ensure planarity. We prove that strongly
  monotone (and thus increasing-chord) drawings of trees and binary
  cactuses require exponential resolution in the worst case, answering
  an open question by Kindermann et al.~\cite{kssw-mdt-14}.  Moreover,
  we provide a binary cactus that does not admit a \sa drawing.  Finally,
  we show that 3-connected planar graphs admit \ic drawings in the
  hyperbolic plane and characterize the trees that admit such
  drawings.
\end{abstract}

\section{Introduction}

Finding paths between two vertices is one of the most fundamental
tasks users want to solve when considering graph drawings~\cite{lppfh-ttgv-06}, for example
to find a connection in a schematic map of a public transport
system. Empirical studies have shown that users perform better in
path-finding tasks if the drawings exhibit a strong geodesic-path
tendency~\cite{heh-grbgt-09,phnk-ulgd-12}. Not surprisingly, graph
drawings in which a path with certain properties exists between every
pair of vertices have become a popular research topic.  Over the last
years a number of different drawing conventions implementing the
notion of strong geodesic-path tendency have been suggested, namely
\emph{greedy drawings}~\cite{Rao2003}, \emph{(strongly) monotone
  drawings}~\cite{acbfp-mdg-2012}, and \emph{self-approaching} and
\emph{increasing-chord drawings}~\cite{acglp-sag-12}.  Note that
throughout this paper, all drawings are straight-line and vertices are
mapped to distinct points.

The notion of greedy drawings came first and was
introduced by Rao et
al.~\cite{Rao2003}.  Motivated by greedy routing schemes, e.g., for
sensor networks, one seeks a drawing where for every pair of vertices
$s$ and $t$, there exists an $st$-path, along which the distance to $t$ decreases in every vertex.
 This ensures that greedily sending a message to a vertex that is
closer to the destination guarantees delivery.
Papadimitriou and Ratajczak conjectured that every 3-connected planar
graph admits a greedy embedding into the Euclidean
plane~\cite{Papadimitriou2005}.  This conjecture has been proved
independently by Leighton and Moitra~\cite{Moitra2008} and Angelini et
al.~\cite{angelini2010algorithm}.  Kleinberg~\cite{k-gruhs-07} showed
that every connected graph has a greedy drawing in the hyperbolic
plane.  Eppstein and Goodrich~\cite{Eppstein2011} showed how to
construct such an embedding, in which the coordinates of each vertex
are represented using only~$O(\log n)$ bits, and Goodrich and
Strash~\cite{Goodrich2009} provided a corresponding \emph{succinct}
representation for greedy embeddings of 3-connected planar graphs
in~$\R^2$. Angelini et al.~\cite{Angelini2012} showed that some graphs
require exponential area for a greedy drawing in~$\mathbb R^2$.  Wang
and He~\cite{Wang2014} used a custom distance metric to construct
planar, convex and succinct greedy embeddings of 3-connected planar
graphs using Schnyder realizers~\cite{Schnyder1990}.  Nöllenburg and
Prutkin~\cite{np-egdt-2013} characterized trees admitting a Euclidean
greedy embedding.  However, a number of interesting questions remain
open, e.g., whether every 3-connected planar graph admits a planar and
convex Euclidean greedy embedding (strong Papadimitriou-Ratajczak
conjecture~\cite{Papadimitriou2005}).  Regarding planar greedy
drawings of triangulations, the only known result is an existential
proof and a heuristic construction by Dhandapani~\cite{Dhandapani2010}
based on face-weighted Schnyder drawings.

While getting closer to the destination in each step, a greedy path can make
numerous turns and may even look like a spiral, which hardly matches
the intuitive notion of geodesic-path tendency.  To overcome this,
Angelini et al.~\cite{acbfp-mdg-2012} introduced monotone drawings,
where one requires that for every pair of vertices~$s$ and~$t$ there
exists a \emph{monotone path}, i.e., a path that is monotone with
respect to some direction.  Ideally, the monotonicity direction should
be~$\vec{st}$.  This property is called \emph{strong monotonicity}.
Angelini et al. showed that biconnected planar graphs admit monotone
drawings~\cite{acbfp-mdg-2012} and that plane graphs admit monotone
drawings with few bends~\cite{adkmrsw-mdgfe-2012}.  Kindermann
et al.~\cite{kssw-mdt-14} showed that every tree admits a strongly
monotone drawing. The existence of strongly monotone planar drawings
remains open, even for triangulations.

Both greedy and monotone paths may have arbitrarily large
\emph{detour}, i.e., the ratio of the path length and the distance of
the endpoints can, in general, not be bounded by a constant.
Motivated by this fact, Alamdari et al.~\cite{acglp-sag-12} recently
initiated the study of \emph{self-approaching} graph drawings.
Self-approaching curves, introduced by Icking~\cite{ikl-sac-99}, are
curves, where for any point $t'$ on the curve, the distance to $t'$
decreases continuously while traversing the curve from the start to
$t'$.  Equivalently, a curve is \sa if, for any three points $a$, $b$,
$c$ in this order along the curve, it is $\dist{a}{c} \geq
\dist{b}{c}$, where $\textrm{dist}$ denotes the Euclidean distance.
An even stricter requirement are so-called \emph{increasing-chord}
curves, which are curves that are self-approaching in both directions.
The name is motivated by the characterization of such curves, which
states that a curve has increasing chords if and only if for any four
distinct points $a,b,c,d$ in that order, it is $\dist{b}{c} \le
\dist{a}{d}$.  Self-approaching curves have detour at most
5.333~\cite{ikl-sac-99} and increasing-chord curves have detour at
most 2.094~\cite{r-cic-94}.  Alamdari et al.~\cite{acglp-sag-12}
studied the problem of recognizing whether a given graph drawing is
self-approaching and gave a complete characterization of trees
admitting self-approaching drawings. Furthermore, Alamdari et
al.~\cite{acglp-sag-12} and Frati et al.~\cite{dfg-icgps-14}
investigated the problem of connecting given points to an \ic drawing.

We note that every \ic drawing is self-approaching and strongly
monotone~\cite{acglp-sag-12}.  The converse is not true.  A self-approaching
drawing is greedy, but not necesserily monotone, and a greedy drawing is
generally neither \sa nor monotone.  For trees, the notions of \sa and \ic drawing coincide since all paths are unique.

\paragraph{Contribution.} 
We obtain the following results on constructing \sa or \ic drawings.
\begin{myenum}
\item We show that every triangulation has an \ic drawing (answering
  an open question of Alamdari et al.~\cite{acglp-sag-12}) and
  construct a \emph{binary cactus} that does not admit a \sa drawing
  (Sect.~\ref{sec:3conn}).  The latter is a notable difference to
  greedy drawings since both constructions of greedy drawings for
  3-connected planar graphs~\cite{Moitra2008,angelini2010algorithm}
  essentially show that every binary cactus has a greedy drawing.  We
  also prove that strongly monotone (and, thus, increasing-chord)
  drawings of trees and binary cactuses require exponential resolution
  in the worst case, answering an open question by Kindermann et
  al.~\cite{kssw-mdt-14}.

\item We show how to construct plane \ic drawings for \emph{planar
    3-trees} (a special class of triangulations) using Schnyder
  realizers (Sect.~\ref{sec:schnyder}).  To the best of our
  knowledge, this is the first construction for this graph class, even
  for greedy and strongly monotone plane drawings, which addresses an open question of Angelini et al.~\cite{acbfp-mdg-2012}.

\item We show that, similar to the greedy case, the hyperbolic plane
  $\Hp$ allows representing a broader class of graphs than~$\R^2$ 
  (Sect.~\ref{sec:hyperbolic}).  We prove that a tree has a \sa or \ic
  drawing in~$\Hp$ \ifao it either has maximum degree~3 or is a
  subdivision of~$K_{1,4}$ (this is not the case in~$\R^2$; see the
  characterization by Alamdari et al.~\cite{acglp-sag-12}),
  implying every 3-connected planar graph has an \ic drawing. We also show how to construct planar \ic drawings of
  binary cactuses in~$\Hp$.
\end{myenum}
\section{Preliminaries}
\label{sec:prelim}

For points $a,b,c,d \in \R^2$, let~$\ray{a}{b}$ denote the ray with origin~$a$
and direction~$\vec{ab}$ and let~$\ray{a}{\vec{bc}}$ denote the ray with origin~$a$ and
direction~$\vec{bc}$. Let~$\dir{ab}$ be the vector~$\vec{ab}$ normalized to unit
length. Let~$\angle(\vec{ab}, \vec{cd})$ denote the smaller angle formed by the two
vectors $\vec{ab}$ and $\vec{cd}$.  
For an angle~$\alpha \in [0, 2\pi]$, let~$R_\alpha$ denote the
rotation matrix $\matr{\cos \alpha}{-\sin \alpha}{\sin \alpha}{\cos
  \alpha}$.
\newcommand{\dirintcont}[1]{\overline{#1}} For vectors~$\vec{v_1}, \vec{v_2}$
with $\dir{\vec{v_2}} = R_\alpha \cdot \dir{\vec{v_1}}$, $\alpha \in [0, 2\pi)$,
we write $\angleccw{\vec{v_1}}{\vec{v_2}} := \alpha$. 

We reuse some notation from the work of Alamdari et
al.~\cite{acglp-sag-12}.  For points $p,q \in \mathbb R^2, p \neq q$,
let~$\hp{p}{q}$ denote the halfplane not containing $p$ bounded by the line
through~$q$ orthogonal to the segment~$pq$.  A piecewise-smooth curve is
self-approaching if and only if for each point~$a$ on the curve, the line
perpendicular to the curve at~$a$ does not intersect the curve at a later
point~\cite{ikl-sac-99}.  This leads to the following characterization of \saps. 

\begin{fact}[Corollary 2 in \cite{acglp-sag-12}]
 \label{rem:sap-char}
 Let $\rho = (v_1, v_2, \ldots, v_k)$ be a directed path embedded in $\R^2$ with
 straight-line segments. Then,~$\rho$ is self-approaching if and only if for all
 $1 \le i<j \leq k$, the point $v_j$ lies in $\hp{v_{i}}{v_{i+1}}$.
\end{fact}

\begin{figure}[tb]
  \hfill \subfloat[]{ \includegraphics[]{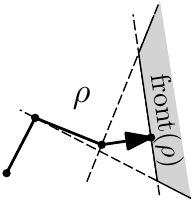} \label{fig:sap-front} }
  \hfill
  \subfloat[]{ \includegraphics[page=1]{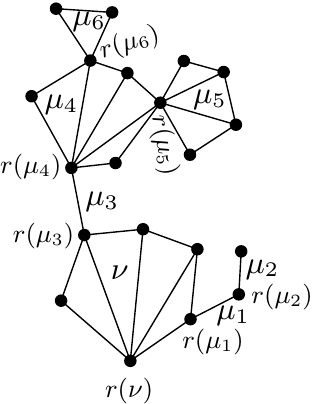} \label{fig:triang-cactus}}
  \hfill
  \subfloat[]{ \includegraphics[page=2]{fig/triang-cactus} \label{fig:triang-cactus:bctree}}
  \hfill\null
  \caption{\protect\subref{fig:sap-front} \sap $\rho$ and $\front(\rho)$ (gray).
    \protect\subref{fig:triang-cactus},
    \protect\subref{fig:triang-cactus:bctree}: downward-triangulated binary cactus and
    the corresponding \bc-tree. $B$-nodes are black, $C$-nodes white. }
  \end{figure}
  \noindent We shall denote the reverse of a path~$\rho$ by~$\rho^{-1}$. 
 Let $\rho =
  (v_1, v_2, \ldots, v_k)$ be a \sap. Define
  $\front(\rho) = \bigcap_{i=1}^{k-1}~\hp{v_i}{v_{i+1}}$, see also Fig.~\ref{fig:sap-front}.
Using Fact~\ref{rem:sap-char}, we can decide
 whether a concatenation of two paths is \sa.
\begin{fact}  Let $\rho_1=(v_1$, \ldots, $v_k)$ and $\rho_2 = (v_k, v_{k+1}$, \ldots, $v_m)$ be
  \sa paths. The path $\concat{\rho_1}{\rho_2}:= (v_1, \ldots, v_k, v_{k+1}$, $\ldots,
  v_m)$ is self-approaching if and only if $\rho_2 \subseteq \front(\rho_1)$.
 \label{rem:contpath}
\end{fact}

A path~$\rho$ has \emph{increasing chords} if for any points $a,b,c,d$ in this
order along~$\rho$, it is $\dist{b}{c} \leq \dist{a}{d}$. A path has increasing
chords \ifao it is \sa in both directions. The following result is easy to see.

\newcommand{\LemPathIcText}{Let~$\rho = (v_1, \dots, v_k)$ be a path such that for any~$i<j$, $i,j \in \{1,
  \dots, k-1\}$, it is $\ang{\vecl{v_i v_{i+1}}}{\vecl{v_j v_{j+1}}} \leq
90\dg$. Then, $\rho$ has increasing chords.}
\begin{lemma} 
\label{lem:path-ic90} 
\LemPathIcText
\end{lemma}
\begin{proof}
  For any~$j > i$, $i,j \in \{1, \dots, k-1\}$, it is
  $\ang{\vecl{v_{j+1} v_j}}{\vecl{v_{i+1} v_i}} \leq 90\dg$. Thus,
  the condition of the lemma also holds for~$\rho^{-1}$, and by
  symmetry it is sufficient to prove that~$\rho$ is \sa.

  We claim that for each $i \in \{ 1, \dots, k-1\}$ and each~$j \in \{
  i+1, \dots, k\}$, it is $v_j \in \hp{v_i}{v_{i+1}}$.  Once the claim
  is proved, it follows from Fact~\ref{rem:sap-char} that $\rho$ is \sa. 
  For the proof of the claim let~$i \in \{ 1, \dots, k-1\}$ be
  arbitrary and fixed.  It suffices to show that $v_{i+2}$, \dots,
  $v_k \in \hp{v_i}{v_{i+1}}$.

  First consider $v_{i+2}$.  By the condition of the lemma, it is
  $\ang{\vec{v_i v_{i+1}}}{\vec{v_{i+1} v_{i+2}}} \leq 90\dg$.
  Therefore, $v_{i+2} \in \hp{v_i}{v_{i+1}}$.  Now assume $v_j \in
  \hp{v_i}{v_{i+1}}$ for some $j \in \{ i+2, \dots, k-1\}$. We show
  $v_{j+1} \in \hp{v_i}{v_{i+1}}$. Consider the halfplane~$h \subseteq
  \hp{v_i}{v_{i+1}}$ whose boundary is parallel to that of
  $\hp{v_i}{v_{i+1}}$ and contains~$v_j$.  Since $\angle(\vec{v_i
    v_{i+1}}, \vec{v_j v_{j+1}}) \leq 90\dg$, it is $v_{j+1} \in h
  \subseteq \hp{v_i}{v_{i+1}}$.  
\end{proof}

\newcommand{\bdepth}[1]{\textnormal{depth}_B(#1)}
\newcommand{\cdepth}[1]{\textnormal{depth}_C(#1)} Let~$G=(V,E)$ be a
connected graph.  A \emph{separating $k$-set} is a set of $k$ vertices whose removal
disconnects the graph.  A vertex forming a separating $1$-set is called
\emph{cutvertex}.  A graph is \emph{$c$-connected} if it does not admit a
separating $k$-set with $k \le c-1$; $2$-connected graphs are also called
\emph{biconnected}.  A connected graph is biconnected if and only if it does not
contain a cutvertex. A \emph{block} is a
maximal biconnected subgraph.  The \emph{block-cutvertex tree} (or
\emph{\bc-tree})~$T_G$ of~$G$ has a \emph{$B$-node} for each block of~$G$, a
\emph{$C$-node} for each cutvertex of $G$ and, for each block~$\nu$ containing a
cutvertex~$v$, an edge between the corresponding $B$- and $C$-node.
We associate $B$-nodes with their corresponding blocks and $C$-nodes with their
corresponding cutvertices.

The following notation follows the work of Angelini et
al.~\cite{angelini2010algorithm}.  Let~$T_G$ be rooted at some block~$\nu$
containing a non-cutvertex (such a block~$\nu$ always exists).  For each block
$\mu \neq \nu$, let~$\pred{\mu}$ denote the \emph{parent block} of~$\mu$, i.e.,
the grandparent of~$\mu$ in~$T_G$. Let~$\gpred{2}{\mu}$ denote the
parent block of~$\pred{\mu}$ and, generally,~$\gpred{i+1}{\mu}$ the parent block of~$\gpred{i}{\mu}$. Further, we define the \emph{root
  $\rt{\mu}$} of $\mu$ as the cutvertex contained in both~$\mu$
and~$\pred{\mu}$. Note that~$\rt{\mu}$ is the parent of~$\mu$ in~$T_G$.  In
addition, for the root node $\nu$ of $T_G$, we define~$\rt{\nu}$ to be some
non-cutvertex of~$\nu$.  Let~$\bdepth{\mu}$ denote the number of $B$-nodes on
the $\nu \mu$-path in~$T_G$ minus~1, and let~$\cdepth{\rt{\mu}} =
\bdepth{\mu}$. If~$\mu$ is a leaf of~$T_G$, we call it a \emph{leaf block}.

A \emph{cactus} is a graph in which every edge is part of at most one
cycle. Note that every cactus is outerplanar. In a \emph{binary}
cactus every cutvertex is part of exactly two blocks.  For a binary
cactus~$G$ with a block~$\mu$ containing a cutvertex $v$,
let~$\subcactb{v}{\mu}$ denote the maximal connected subgraph
containing~$v$ but no other vertex of~$\mu$.  We say
that~$\subcactb{v}{\mu}$ is a \emph{subcactus} of~$G$.
For a fixed cactus root, the block~$\mu$ containing~$v$, such that $v
\neq \rt{\mu}$, is unique, and we write~$\subcact{v}$
for~$\subcactb{v}{\mu}$.

A \emph{triangulated} cactus is a cactus together with additional
edges, which make each of the cactus blocks internally
triangulated. A \emph{triangular fan with vertices
  $V_t= \{v_0, v_1, \dots, v_k\}$ and root~$v_0$} is a graph on~$V_t$
with edges $v_i v_{i+1}$, $i=1, \dots, k-1$, as well as~$v_0 v_i$,
$i=1, \dots, k$. Let us consider a special kind of triangulated
cactuses, each of whose blocks~$\mu$ is a triangular fan with root~$\rt{\mu}$. 
We call such a cactus \emph{downward-triangulated} and every edge of a
block~$\mu$ incident to~$\rt{\mu}$ a \emph{downward}
edge. Fig.~\ref{fig:triang-cactus} and~\ref{fig:triang-cactus:bctree} show a
downward-triangulated binary cactus and the corresponding \bc-tree.

For a fixed straight-line drawing of a binary cactus~$G$, we define
the set of \emph{upward} directions~$\upwdir{G} = \{ \vec{\rt{\mu} v}
\mid \mu \textnormal{ is a block of } G \textnormal{ containing } v,
\; v \neq \rt{\mu} \}$ and the set of \emph{downward}
directions~$\downwdir{G} = \{ \vec{uv} \mid \vec{vu} \in \upwdir{G}
\}$. 

\section{Graphs with Self-Approaching Drawings}
\label{sec:3conn}

A natural approach to construct (not necessarily plane) \sads is to construct a
\sad of a spanning subgraph. For instance, to draw a graph~$G$ containing a
Hamiltonian path~$H$ with increasing chords, we simply draw~$H$ consecutively on
a line. In this section, we consider 3-connected planar graphs and the special
case of triangulations, which addresses an open question of Alamdari et al.~\cite{acglp-sag-12}. These graphs are known to have a spanning binary
cactus~\cite{angelini2010algorithm, Moitra2008}. Angelini et
al.~\cite{angelini2010algorithm} showed that every triangulation has a spanning
downward-triangulated binary cactus.

\subsection{Increasing-chord drawings of triangulations}
\label{subsec:ic-triang}
We show that every downward-triangulated binary
cactus has an \ic drawing.
The construction is similar to the one of the greedy
drawings of binary cactuses in the two proofs of the Papadimitriou-Ratajczak
conjecture~\cite{Moitra2008, angelini2010algorithm}.  Our proof is by induction
on the height of the \bc-tree. We show that $G$ can be drawn
such that all downward edges are almost vertical and the
remaining edges almost horizontal. Then, for vertices~$s,t$ of~$G$, an
$st$-path with increasing chords goes downwards to some block~$\mu$,
then sideways to another cutvertex of~$\mu$ and, finally, upwards
to~$t$. Let~$\vright,\vup$ be vectors $\colvect{1}{0}$, $\colvect{0}{1}$. 

\begin{theorem}
  \label{lem:triang-ic}
  Let~$G=(V,E)$ be a downward-triangulated binary cactus. For any~$0\dg < \eps <
  90\dg$, there exists an increasing-chord drawing~$\Gamma_\eps$ of~$G$, such
  that for each vertex~$v$ contained in some block~$\mu$, $v \neq r(\mu)$, the
  angle formed by $\vec{r(\mu) v}$ and~$\vup$ is at
  most~$\frac{\eps}{2}$. 
\end{theorem}
\begin{proof}
\begin{figure}[tb] \hfill
\subfloat[]{ \includegraphics[page=1]{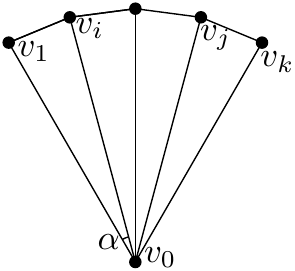}
    \label{fig:triang-induction:1}}\hfill
\subfloat[]{ \includegraphics[page=2]{fig/triang-induction.pdf}
    \label{fig:triang-induction:2}}\hfill
\subfloat[]{ \includegraphics[page=3]{fig/triang-induction.pdf}\hfill\null
  \label{fig:triang-induction:safe-region}}\hfill\null
\caption{Drawing a triangulated binary cactus with increasing chords
  inductively. The drawings~$\Gamma_{i,\eps'}$ of the subcactuses,
  $\eps' = \frac{\eps}{4k}$, are contained inside the gray cones. It
  is $\beta = 90\dg - \eps'$, $\gamma = 90\dg + \eps'/2$.}
\end{figure}
Let~$G$ be rooted at block~$\nu$. As our base case, let $\nu=G$
be a triangular fan with vertices~$v_0, v_1, \ldots, v_k$ and root~$v_0 =
r(\nu)$. We draw~$v_0$ at the origin and distribute $v_1$, \ldots, $v_k$ on the
unit circle, such that $\ang{\vup}{\vec{v_0 v_1}}= k \alpha/2$ and
$\ang{\vec{v_0 v_i}}{\vec{v_0 v_{i+1}}} = \alpha$, $\alpha = \eps/k$; see
Fig.~\ref{fig:triang-induction:1}. By Lemma~\ref{lem:path-ic90}, path~$(v_1,
\ldots, v_k)$ has increasing chords.

Now let~$G$ have multiple blocks. We draw the root block~$\nu$, $v_0 = r(\nu)$,
as in the previous case, but with $\alpha = \frac{\eps}{2k}$. Then, for each $i
= 1, \ldots, k$, we choose $\eps' = \frac{\eps}{4k}$ and draw the
subcactus~$G_i = G_\nu^{v_i}$ rooted at~$v_i$ inductively, such that the corresponding
drawing~$\Gamma_{i, \eps'}$ is aligned at $\vec{v_0 v_i}$ instead of~$\vup$; see
Fig~\ref{fig:triang-induction:2}. 
Note that~$\eps'$ is the angle of the cones (gray) containing~$\Gamma_{i, \eps'}$. 
Obviously, all downward edges of $G$ form angles at
most~$\frac{\eps}{2}$ with~$\vup$.

We must be able to reach any $t$ in any~$G_j$ from any~$s$ in
any~$G_i$ via an \ic path~$\rho$.  To achieve this, we make sure that
no normal on a downward edge of~$G_i$ crosses the drawing of~$G_j$, $j
\neq i$. Let~$\Lambda_i$ be the cone with apex~$v_i$ and angle~$\eps'$
aligned with $\vec{v_0 v_i}$, $ v_0 \not \in \Lambda_i$ (gray regions
in Fig.~\ref{fig:triang-induction:2}). Let~$s_i^l$ and~$s_i^r$ be the
left and right boundary rays of~$\Lambda_i$ with respect to~$\vec{v_0
  v_i}$, and~$h_i^l$, $h_i^r$ the halfplanes with boundaries
containing~$v_i$ and orthogonal to $s_i^l$ and~$s_i^r$ respectively,
such that $v_0 \in h_i^l \cap h_i^r$.   Define $\Diamond_i = \Lambda_i
\cap h_{i-1}^r \cap h_{i+1}^l$ (thin blue quadrilateral in
Fig.~\ref{fig:triang-induction:safe-region}), and
analogously~$\Diamond_j$ for~$j \neq i$. It holds $\Diamond_j
\subseteq h_{i}^r \cap h_{i}^l$ for each $i \ne j$.  We now scale each drawing
$\Gamma_{i,\eps'}$ such that it is contained in $\Diamond_i$.  In particular, for any downward edge $uv$
in~$\Gamma_{i, \eps'}$, we have $\Gamma_{j,\eps'} \subseteq \Diamond_j \subseteq \hp{u}{v}$ for $j \ne i$.
We claim that the resulting drawing of $G$ is an increasing-chord drawing.

Consider vertices~$s$,$t$ of~$G$.
If $s$ and $t$ are contained in the same subgraph $G_i$, an \ic $st$-path in $G_i$ exists by induction.
If $s$ is in $G_i$ and $t$ is $v_0$, let~$\rho_i$ be the $s v_i$-path in~$G_i$ that
uses only downward edges. By Lemma~\ref{lem:path-ic90}, path $\rho_i$ is \ic and remains so after adding edge~$v_i v_0$.

Finally, assume $t$ is in $G_j$ with $j \ne i$.
Let~$\rho_j$ be the
$t v_j$-path in~$G_j$ that uses only downward edges. Due to the choice
of~$\eps'$, $h_{i}^r \cap h_{i}^l \subseteq \front(\rho_i)$ contains $v_1,
\ldots, v_k$ in its interior. Consider the path $\rho' = (v_i, v_{i+1}, \ldots,
v_j)$. It is \sa by Lemma~\ref{lem:path-ic90}; also, $\rho' \subseteq \front(\rho_i)$
and~$\rho_j \subseteq \front(\rho')$. It also holds~$\rho_j \subseteq \Diamond_j
\subseteq \front(\rho_i)$. Fact~\ref{rem:contpath} lets us
concatenate $\rho_i$, $\rho'$ and~$\rho_j^{-1}$ to a \sa path. By a symmetric argument, it is also \sa 
in the opposite direction and, thus, is \ic. 
\end{proof}

Since every triangulation has a spanning downward-triangulated binary
cactus~\cite{angelini2010algorithm}, this implies that planar
triangulations admit increasing-chord drawings.

\begin{corollary}
  \label{cor:triangulation-increasing-chord}
  Every planar triangulation admits an increasing-chord drawing.
\end{corollary}

\subsection{Exponential worst case resolution}
\newcommand{\dircone}[1]{U_{#1}}
The construction for a spanning downward-triangulated binary cactus in
Section~\ref{subsec:ic-triang} requires exponential area. In this
section, we show that we cannot do better in the worst case even for
strongly monotone drawings of downward-triangulated binary
cactuses. Recall that \ic drawings are strongly monotone.

The following lemma describes directions of certain edges in a greedy
or monotone drawing of a cactus.

  \begin{figure}[tb]
    \hfill
    \subfloat[]{ \includegraphics[page=2]{fig/dir-lemma} \label{fig:diverge}}\hfill
    \subfloat[]{\includegraphics[page=1]{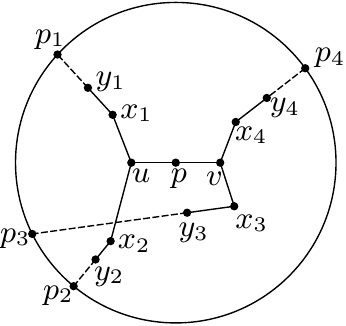} \label{fig:slope-disj}}
    \hfill\null
    \caption{\protect\subref{fig:diverge}~Proof of
      Lemma~\ref{lem:diverge}; \protect\subref{fig:slope-disj}~proof of
      Lemma~\ref{lem:cact:slope-disj}.}
  \end{figure}

\begin{figure}[tb]
  \hfill
  \subfloat[]{\includegraphics[page=1]{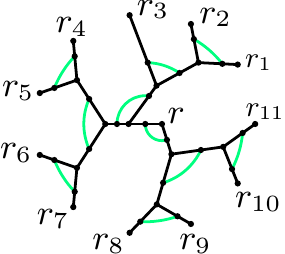} \label{fig:strmon-cact-exp:centraltree}}
  \hfill
  \subfloat[]{\includegraphics[page=2]{fig/strmon-cact-exp.pdf} \label{fig:strmon-cact-exp:caterpillar}}
  \hfill
  \subfloat[]{\includegraphics[page=3]{fig/strmon-cact-exp.pdf} \label{fig:strmon-cact-exp:angles:1}}
  \hfill
  \subfloat[]{\includegraphics[page=4]{fig/strmon-cact-exp.pdf} \label{fig:strmon-cact-exp:angles:2}}
  \hfill\null
  \caption{Family of binary cactuses~$G_k$ requiring exponential area
    for any strongly monotone
    drawing. \protect\subref{fig:strmon-cact-exp:centraltree}~Central
    cactus~$G'$;
    \protect\subref{fig:strmon-cact-exp:caterpillar}~binary subcactus~$C_k$
    attached to each vertex of degree~1 of~$G'$.  In a strongly
    monotone drawing of~$G_k$, it must hold:
    \protect\subref{fig:strmon-cact-exp:angles:1}~$|u_2 u_4| \leq |u_2
    v_2| \tan \eps$;
    \protect\subref{fig:strmon-cact-exp:angles:2}~$|u_4 v_4| \leq |u_2 u_4|
    \tan \eps$. }
\end{figure}

\newcommand{\LemDivergeText}{For a cactus~$G=(V,E)$ and two vertices $s,t \in V$,
  consider the cutvertices~$v_1$, \dots, $v_k$ lying on every $s t$-path
  in~$G$ in this order. In any greedy drawing of~$G$, the path
  $(s,v_1, \dots, v_k,t)$ is drawn greedily, i.e., each subpath of
  $(s,v_1$, $\ldots$, $v_k, t)$ is greedy. In any monotone drawing,
  the path $(s,v_1, \dots, v_k,t)$ is monotone. In both cases,
  $\ray{v_1}{s}$ and~$\ray{v_{k}}{t}$ diverge.}
\begin{lemma}
  \label{lem:diverge}
  \LemDivergeText 
\end{lemma}
\begin{proof}
  Let~$v_0 = s$, $v_{k+1} = t$. For $0 \leq i < j \leq k+1$, any $v_i
  v_j$ and~$v_j v_i$-path in~$G$ contains vertices~$v_i$, $v_{i+1}$,
  \dots, $v_j$. Since a path in a greedy drawing of~$G$ remains greedy after replacing
  subpaths by shortcuts, the segments~$s v_1$,
  $v_1 v_2$, \ldots, $v_{k-1} v_k$, $v_k v_t$ form a greedy
  drawing. By Lemma~7 of Angelini et al.~\cite{Angelini2012}, $\ray{v_1}{s}$
  and~$\ray{v_{k}}{t}$ diverge; see
  Fig.~\ref{fig:diverge}. 

  Analogously, a path remains monotone after replacing  subpaths by
  shortcuts. Therefore, in a monotone drawing of~$G$, segments~$s
  v_1$, $v_1 v_2$, \ldots, $v_{k-1} v_k$, $v_k v_t$ form a monotone
  drawing. Since a monotone path cannot make a turn of~$180\dg$ or
  more, $\ray{v_1}{s}$ and~$\ray{v_{k}}{t}$ must diverge. 
\end{proof}

\newcommand{\upwedges}[1]{E_U(#1)}
For the following lemma, consider a greedy or monotone drawing of a
cactus~$G$ with root~$r$.  We define the set of \emph{upward directed}
edges $\upwedges{G} = \{ \rt{\mu} v \mid \mu \textnormal{ is a block
  of } G \allowbreak\textnormal{containing } v, \; v \neq \rt{\mu} \}$.  Note
that if $G$ is not triangulated, some edges in~$\upwedges{G}$ might not be
edges in~$G$.
For cutvertex $u$, let $\dircone{u}$ denote the upward directed edges
of the subcactus rooted at $u$ or, formally, $\dircone{u} =
\upwedges{\subcact{u}}$. Then, the following property
holds. 

\begin{lemma}
  \label{lem:cact:slope-disj}
  In a monotone or greedy drawing of a cactus with root~$r$, consider
  cutvertices $u,v \neq r$, such that the subcactuses $\subcact{u}$
  and~$\subcact{v}$ are disjoint.   Then the edges in $\dircone{u}$ and in $\dircone{v}$ each form a single interval in the circular order induced by their directions.
\end{lemma}
\begin{proof}
  Consider four pairs of vertices $x_i$, $y_i$, $i = 1, \dots, 4$,
  such that $x_i y_i \in \dircone{u}$ for $i = 1,2$ and $x_i y_i \in
  \dircone{v}$ for $i=3,4$. For~$i=1,2$, $j=3,4$, let~$\rho_{ij}$
  denote the vertex sequence $y_i$, $x_i$, $u$, $v$, $x_j$,
  $y_j$. Since~$x_i$, $u$, $v$, $x_j$ are cutvertices, $\rho_{ij}$ is
  a subsequence of every $y_i y_j$ path. Therefore, each
  such~$\rho_{ij}$ forms a monotone or a greedy drawing of a path,
  respectively. In both cases, $\rho_{ij}$ is non-crossing and cannot
  make a turn of~$180\dg$ or more. Therefore, rays $\ray{x_i}{y_i}$
  and~$\ray{x_j}{y_j}$ must diverge. Finally, neither $\ray{x_i}{y_i}$
  nor $\ray{x_j}{y_j}$ can cross~$\rho_{ij}$.

  We define $p=(u+v)/2$ and choose an arbitrary $R>0$, such that all
  paths~$\rho_{ij}$ are contained inside a circle $C$ with
  center~$p$ and radius~$R$. Let~$p_i$ be the intersection
  of~$\ray{x_i}{\vec{x_i y_i}}$ and~$C$. Assume $p_1$,
  $p_3$, $p_2$, $p_4$ is the counterclockwise order of~$p_i$ on the
  boundary of~$C$; see Fig.~\ref{fig:slope-disj}. Then, for some
  pair~$i,j$, $i \in \{ 1,2 \}$, $j \in \{ 3,4 \}$, there exists a
  crossing of~$\ray{x_i}{y_i}$ or~$\ray{x_j}{y_j}$
  with~$\rho_{ij}$ or with each other; a contradiction. Therefore,
  $p_1$, $p_2$ as well as $p_3$, $p_4$ appear consecutively on the
  boundary of~$C$. For~$R \rightarrow \infty$, $\vec{p p_i}$ becomes
  parallel with $\vec{x_i y_i}$. Therefore, the circular order
  $\vec{x_1 y_1}$, $\vec{x_3 y_3}$, $\vec{x_2 y_2}$, $\vec{x_4 y_4}$
  is not possible, and the statement follows.
\end{proof}

Note that for trees, Angelini et al.~\cite{acbfp-mdg-2012} call this property \emph{slope disjointness}.
Consider the following family of binary cactuses $G_k$.  Let $G'$ be a
rooted binary cactus with exactly eleven vertices~$r_1, \ldots, r_{11}$ of degree~1
and its root~$r$ as the only vertex of degree~2; see
Fig.~\ref{fig:strmon-cact-exp:centraltree}. 
Next, consider cactus~$C_k$ consisting of a chain of~$k$ triangles and
some additional degree-1 nodes as in
Fig.~\ref{fig:strmon-cact-exp:caterpillar}.
We construct~$G_k$ by attaching a copy of~$C_k$ to each~$r_i$
in~$G'$. From now on, consider a strongly monotone
drawing of~$G_k$. 

Using Lemma~\ref{lem:cact:slope-disj} and the pigeonhole
principle, we can show the following fact.

\begin{lemma}
  \label{lem:strmon-cact:narrow-cone}
  For some $r_i$, $i \in \{ 1, \dots, 11 \}$, each pair of directions
  in $\dircone{r_i}$ forms an angle at most $\eps = 360\dg /
  11$. 
\end{lemma}
\begin{proof}
  Consider the two cutvertices of the root block of~$G_k$; see
  Fig.~\ref{fig:strmon-cact-exp:centraltree}. By
  Lemma~\ref{lem:cact:slope-disj}, vectors in $\dircone{r_1} \cup
  \dots \cup \dircone{r_{11}}$ appear in the following circular order:
  first the vectors in $\dircone{r_1} \cup \dots \cup
  \dircone{r_{7}}$, then the vectors in $\dircone{r_8} \cup \dots \cup
  \dircone{r_{11}}$. By applying the same argument to the child blocks
  repetitively, it follows that the vectors have the following
  circular order: first the vectors in $\dircone{r_{\pi(1)}}$, then
  the vectors in~$\dircone{r_{\pi(2)}}$, \dots, then the vectors in
  $\dircone{r_{\pi(11)}}$ for some permutation~$\pi$. Therefore, for
  some $i$, each pair of directions in $\dircone{r_i}$ form an angle
  at most $\eps = 360\dg / 11$.
\end{proof}

Now consider a vertex $r_i$ with the property of
Lemma~\ref{lem:strmon-cact:narrow-cone}. Let the vertices of its
subcactus be named as in Fig.~\ref{fig:strmon-cact-exp:caterpillar}.
\Wlog, we may assume that each vector in~$\dircone{r_i}$ forms an
angle at most~$\eps/2$ with the upward direction~$\vup$. We show that
certain directions have to be almost horizontal.

\begin{lemma}
  \label{lem:strmon-cact:horiz}
  For even $i,j$, $2 \leq j \leq i$, consider vertices~$u_i$,
  $v_j$. Vector~$\vec{u_i v_j}$ forms an angle at most~$\eps/2$ with
  the horizontal axis.
\end{lemma}
\begin{proof}
  Consider a strongly monotone~$u_i v_j$ path~$\rho$. Vertices~$u_i$,
  $u_{i-1}$, $v_{j-1}$, $v_{j}$ must appear on~$\rho$ in this
  order. It is $\angle(\vec{u_{i-1} u_i}, \vec{v_{j-1}, v_j}) \leq
  \eps$. Furthermore, by the strong monotonicity of~$\rho$, it is
  $\angle u_{i-1} u_i v_j$, $\angle v_{j-1} v_j u_i <
  90\dg$. Therefore, $\angle u_{i-1} u_i v_j$, $\angle v_{j-1} v_j u_i
  \in(90\dg-\eps, 90\dg)$, and the statement follows.
\end{proof}

The following lemma essentially shows that $G_k$ requires exponential
resolution.

\begin{lemma}
  \label{lem:strmon-cact:lengths}
  For $i =2, 4, \dots, 2k$, it holds $|u_{i+2}v_{i+2}| \leq (\tan \eps)^2 |u_{i}v_{i}|$.
\end{lemma}
\begin{proof}
  For brevity, let $i=2$. First, we show that $|u_2 v_2|$ is
  significantly bigger than $|u_2 u_4|$; see
  Fig.~\ref{fig:strmon-cact-exp:angles:1}. It holds:
  \begin{equation*}
 \frac{|u_2 u_4|}{|u_2 v_2|} = \frac{\sin \angle u_2 v_2 u_4}{\sin \angle u_2 u_4 v_2} \leq \frac{\sin \eps}{\sin(90\dg - \eps)} = \tan \eps.
\end{equation*}

Next, we show that $|u_2 u_4|$ is
  significantly bigger than $|u_4 v_4|$; see
  Fig.~\ref{fig:strmon-cact-exp:angles:2}. It holds:
  \begin{equation*}
 \frac{|u_4 v_4|}{|u_2 u_4|} = \frac{\sin \angle u_4 u_2 v_4}{\sin \angle u_2 v_4 u_4 } \leq \frac{\sin \eps}{\sin(90\dg - \eps)} = \tan \eps.
\end{equation*}
Thus, $|u_4 v_4| \leq |u_2 u_4| \tan \eps \leq |u_2 v_2| (\tan \eps)^2$.
\end{proof}

As a consequence of Lemma~\ref{lem:strmon-cact:lengths} we get $|u_{2k+2} v_{2k+2}| \leq |u_2 v_2| (\tan
\eps)^{2 k}$. It is $(\tan \eps)^2 < 0.414$. Since cactus~$G_k$ contains
$n = \Theta(1) + 44k$ vertices, the following exponential lower bound holds for
the resolution of strongly monotone drawings.

\begin{theorem}
  \label{thm:strmon-cactus:resolution}
  There exists an infinite family of binary cactuses with $n$ vertices that require
  resolution~$\Omega(2^\frac{n}{44})$ for any strongly monotone
  drawing.
\end{theorem}

Using this result, we can construct a family of trees requiring
exponential area for any strongly monotone drawing. Consider the
binary spanning tree~$T_k$ of~$G_k$ created by removing the green
edges in Fig.~\ref{fig:strmon-cact-exp:centraltree}
and~\ref{fig:strmon-cact-exp:caterpillar}. Obviously, by
Theorem~\ref{thm:strmon-cactus:resolution} it requires
resolution~$\Omega(2^\frac{n}{44})$ for any strongly monotone
drawing. This answers an open question by Kindermann et
al.~\cite{kssw-mdt-14}.  Replacing degree-2 vertices by shortcuts and
applying a more careful analysis lets us prove the following result.

\begin{theorem}
  \label{thm:strmon-tree:resolution}
  There exists an infinite family of binary trees with $n$ vertices that require
  resolution~$\Omega(2^\frac{n}{22})$ for any strongly monotone
  drawing.
\end{theorem}

\subsection{Non-triangulated cactuses}
The construction for an \ic drawing from Section~\ref{subsec:ic-triang} fails if the
blocks are not triangular fans since we now cannot just use downward
edges to reach the common ancestor block. Consider the family of
rooted binary cactuses $G_n = (V_n, E_n)$ defined as
follows. Graph~$G_0$ is a single 4-cycle, where an arbitrary vertex is
designated as the root.  For~$n \geq 1$, consider two disjoint copies
of~$G_{n-1}$ with roots~$a_0$ and~$c_0$. We create~$G_n$ by adding new
vertices~$r_0$ and~$b_0$ both adjacent to~$a_0$ and~$c_0$; see
Fig.~\ref{fig:square-cactus}. For the new block~$\nu$ containing $r_0,
a_0, b_0, c_0$, we set~$\rt{\nu} = r_0$. We select~$r_0$ as the root
of~$G_n$ and~$\nu$ as its root block. For a block~$\mu_i$ with
root~$r_i$, let $a_i,b_i,c_i$ be its remaining vertices, such
that~$b_i r_i \notin E_n$.  For a given drawing, due to the symmetry
of~$G_n$, we can rename the vertices~$a_i$ and~$c_i$ such that
$\angleccw{\vec{r_i c_i}}{\vec{r_i a_i}} \leq 180\dg$. We now prove
the following negative result.
\begin{theorem}
  \label{thm:no-sad}
  For~$n \geq 10$, $G_n$ has no \sa drawing.
\end{theorem}

\begin{figure}[tb]
  \hfill
  \subfloat[]{ \includegraphics[page=1]{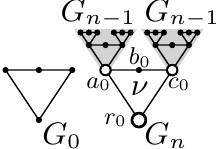} \label{fig:square-cactus}}\hfill
    \subfloat[]{\includegraphics[page=1]{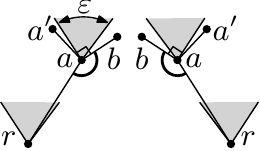}\label{fig:angles-rab}}
\hfill
    \subfloat[]{\includegraphics[page=4]{fig/counterex-basic-lemmas.pdf}\label{fig:angles-horizontal}}
\hfill
    \subfloat[]{\includegraphics[page=5]{fig/counterex-basic-lemmas.pdf}\label{fig:angles-horizontal:2}}
    \hfill\null
  \caption{ \protect\subref{fig:square-cactus} cactuses $G_n$;
        \protect\subref{fig:angles-rab} Lemma~\ref{lem:angles}\sublem{3};
      \protect\subref{fig:angles-horizontal},\protect\subref{fig:angles-horizontal:2}:~Lemma~\ref{lem:angles}\sublem{4}.}
\end{figure}

The outline of the proof is as follows. We show that every \sa drawing~$\Gamma$
of~$G_{10}$ contains a \sa drawing of~$G_3$ 
such that for each block~$\mu$ of~this
$G_3$, the angle at~$r = \rt{\mu}$ is very small, angles at~$a$ and~$c$ are~$90\dg$ or slightly
bigger (Lemma~\ref{lem:angles}) and such that sides~$r a$ and~$r c$ have
almost same length which is significantly greater than~$\dist{a}{c}$
(Lemma~\ref{lem:lengths}).  In addition, the following properties hold for
this~$G_3$. 

\begin{myenum}
\item If $\mu_i$ is contained in the subcactus rooted at~$c_j$, each \sa
  $b_i a_j$-path uses edge~$b_i a_i$, and analogously for the symmetric case;
  see Lemma~\ref{lem:rab-geq90}.
\item Each block is drawn significantly smaller than its parent block; see
  Lemma~\ref{lem:cone-strip}\sublem{1}.
\item If the descendants of block~$\mu$ form subcactuses~$G_k$ with $k \geq 2$ on both
  sides, the parent block of~$\mu$ must be drawn smaller than~$\mu$; see
  Lemma~\ref{lem:cone-strip}\sublem{2}.
\end{myenum}

Obviously, the second and third conditions are contradictory. 
Note that every block has to be \sa. However, it might be non-convex and even
non-planar.
 \newcommand{\ObsSaPolygonText}{ In a \sa drawing of a polygon~$P$, no two
   non-consec\-utive angles can be both less than $90\dg$.}
 \begin{observation}
  \ObsSaPolygonText
  \label{obs:sapolygon}
 \end{observation}
\begin{proof}
  Note that each triangle is \sa.  Let~$v_1, v_2, v_3, v_4$ be pairwise distinct
  vertices appearing in this circular order around the boundary of~$P$.
   Let the angles at both~$v_2$
  and~$v_4$ be less than~$90\dg$. However, a \sa $v_1 v_3$-path must
  use either~$v_2$ or~$v_4$, a contradiction. 
\end{proof}

 \begin{figure}[tb]
    \hfill
  \subfloat[]{ \includegraphics[page=2]{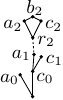} \label{fig:rab-geq90}}\hfill
  \subfloat[]{ \includegraphics[page=6]{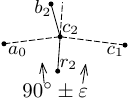} \label{fig:rab-geq90:2}}\hfill
  \subfloat[]{ \includegraphics[page=3]{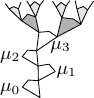} \label{fig:Gr}}
  \hfill\null
  \caption{
    \protect\subref{fig:rab-geq90},\protect\subref{fig:rab-geq90:2}
    construction for Lemma~\ref{lem:rab-geq90};
    \protect\subref{fig:Gr} subcactus~$G_6$ providing the
    contradiction in the proof of Theorem~\ref{thm:no-sad}.}
    \label{fig:counterex-lemmas}
  \end{figure}

The following lemmas will be used to show that the drawings of certain
blocks must be relatively thin, i.e., their downward edges have similar
directions.

\newcommand{\LemThinSubcactusText} {Every \sa drawing of~$G_{10}$
  contains a cutvertex~$\bar{r}$, such that each pair of directions in
  $\dircone{\bar{r}}$ form an angle at most $\eps = 22.5\dg$.  }

\begin{lemma}
  \label{lem:thin-subcactus}
  \LemThinSubcactusText
\end{lemma}
\begin{proof}
  Denote by $r_j$, $j=1, \ldots, 16$, the cutvertices
  with~$\cdepth{r_j} = 4$.  By an argument similar to the one in the
  proof of Lemma~\ref{lem:strmon-cact:narrow-cone}, the edges in
  $\dircone{r_j}$ appear in the following circular order by their
  directions: first the edges in $\dircone{r_{\pi(1)}}$, then the
  edges in $\dircone{r_{\pi(2)}}$, \dots, then the edges in
  $\dircone{r_{\pi(16)}}$ for some permutation~$\pi$.  Therefore, by
  the pigeonhole principle, the statement holds for some~$j \in \{ 1,
  \dots, 16\}$ and~$\bar{r} = r_j$.  
\end{proof}

Let $\bar{r}$ be a cutvertex in
the fixed drawing at $\cdepth{\bar{r}} = 4$ with the property shown in Lemma~\ref{lem:thin-subcactus}. Then, $\subcact{\bar{r}}$ is isomorphic
to~$G_6$. From now on, we only consider non-leaf blocks~$\mu_i$ and
vertices~$r_i, a_i, b_i,c_i$ in~$\subcact{\bar{r}}$. We shall
sometimes name the points~$a$ instead of~$a_i$ etc. for
convenience. We assume $\ang{\vup}{\vec{ra}}$, $\ang{\vup}{\vec{rc}}
\leq \eps/2$. The following lemma is proved using basic trigonometric
arguments.

\newcommand{\LemAnglesText}{ It holds: \emph{\sublem{1}}~$\angle a b c \geq
  90\dg$; \emph{\sublem{2}}~$\subcact{a} \subseteq \hp{b}{a}$, $\subcact{c}
  \subseteq \hp{b}{c}$; \emph{\sublem{3}}~$\angle b a r \leq 90\dg + \eps$,
  $\angle b c r \leq 90\dg + \eps$. \emph{\sublem{4}}~For vertices $u$ in
  $\subcact{a}$, $v$ in $\subcact{c}$ of degree~4 it is $\angle(\vec{uv},
  \vright) \leq \eps/2$.   }
\begin{lemma}
  It holds: 
\begin{myenum}
\item[\emph{\sublem{1}}]~$\angle a b c \geq
  90\dg$;
\item[\emph{\sublem{2}}]~$\subcact{a} \subseteq \hp{b}{a}$,
  $\subcact{c} \subseteq \hp{b}{c}$;
\item[\emph{\sublem{3}}]~$\angle b a r \leq 90\dg + \eps$, $\angle b c
  r \leq 90\dg + \eps$. 
\item[\emph{\sublem{4}}]~For vertices $u$ in $\subcact{a}$, $v$ in
  $\subcact{c}$ of degree~4 it is $\angle(\vec{uv}, \vright) \leq
  \eps/2$.
\end{myenum} 
\label{lem:angles}
\end{lemma}
\begin{proof}
\sublem{1}~It is~$\angle arc \leq \eps$. Thus, by
  Observation~\ref{obs:sapolygon}, $\angle a b c \geq 90\dg$.
\smallskip

\noindent\sublem{2}~Let $t$ be a vertex of~$\subcact{c}$. Since $\angle
  arc \leq \eps$, any \sa $at$-path must contain~$b c$. Thus, $t \in
  \hp{b}{c}$, and the claim for $\subcact{c}$ and, similarly,
  for~$\subcact{a}$ follows.
\smallskip

\noindent\sublem{3}~Consider block~$\mu'$ containing~$a' \neq a$,
  $\rt{\mu'} = a$; see Fig.~\ref{fig:angles-rab}. Then,
  $\angle(\vec{ra}, \vec{a a'}) \leq \eps$. By~\sublem{2}, it is~$b a
  a' \geq 90\dg$. If~$\angle b a r > 90\dg + \eps$, it
  is~$\angle(\vec{ra}, \vec{a a'}) > \eps$, a contradiction. The same
  argument applies for $\angle b c r$. \smallskip

\noindent\sublem{4}~Since~$u,v$ have degree~4, they are roots of some
  blocks. Let~$u_1$ be a neighbor of~$u$ in~$\subcact{u}$ and~$v_1$ a
  neighbor of~$v$ in~$\subcact{v}$ maximizing~$\angle u_1 u v$
  and~$\angle v_1 v u$; see Fig.~\ref{fig:angles-horizontal}. By
  considering \sa $u_1 v$ and $v_1 u$-paths, it follows $\angle u_1 u
  v, \angle v_1 v u \geq 90\dg$. Also, $\ray{u_1}{u}$
  and~$\ray{v_1}{v}$ converge by Lemma~\ref{lem:diverge}. Let~$p$ be
  their intersection. Then, $\angle upv \leq \eps$ and $\angle puv,
  \angle pvu \leq 90\dg$. It is $\angle(\vec{pu}, \vup) \leq \eps/2$
  and~$\angle(\vec{pv}, \vup) \leq \eps/2$. Therefore, if~$\vec{uv}$
  points upward, it forms an angle at most $\eps/2$ with the
  horizontal direction. If $\vec{uv}$ points downward, by symmetric
  arguments, $\vec{vu}$ forms an angle at most $\eps/2$ with the
  horizontal direction. The same holds for~$\vec{ac}$, $\vec{av}$,
  $\vec{uc}$.

  It remains to show that~$u$ is ``to the left''
  of~$v$. Since~$\vec{ra}$ is counterclockwise relative to~$\vec{rc}$,
  it is~$\angle(\vec{ac},\vright) \leq \eps/2$. Assume
  $\angle(\vec{uv},-\vright) \leq \eps/2$. Then, $u$ or~$v$ (\wlg,
  $u$) must be contained in both vertically aligned cones with
  apices~$a$ and~$c$ and angle~$\eps$ (dark gray area in
  Fig.~\ref{fig:angles-horizontal:2}). This contradicts the fact that
  $\vec{uc}$ forms an angle of at most $\eps/2$ with the horizontal
  direction.
\end{proof}

\newcommand{\LemLengthsText}{ It holds: \emph{\sublem{1}}~$ \frac{|ra|}{|rc|}$,
  $\frac{|rc|}{|ra|} \geq \cos \eps$; \emph{\sublem{2}}~$ \frac{|ac|}{|ra|},
  \frac{|ac|}{|rc|} \leq \tan \eps$; \emph{\sublem{3}}~The distance from~$a$ to the
  line through~$rc$ is at least $|ac| \cos \eps$.  \emph{\sublem{4}}~Consider
  block~$\mu$ containing $a,b,c,r$, vertex~$u \neq a$ in~$\subcact{a}$ and~$v
  \neq c$ in~$\subcact{c}$, $\deg(u) = \deg(v)=4$.  Then, $\frac{|au|}{|ac|} \leq
  \tan \eps$ and~$|uv| \leq (1 + 2\tan \eps) |ac|$. }

We can now describe block angles at~$a_i$, $c_i$ more precisely and
characterize certain \sa paths in~$\subcact{\bar{r}}$.  We show that a
\sa path from~$b_i$ \emph{downwards and to the left}, i.e., to an
ancestor block $\mu_j$ of~$\mu_i$, such that~$\mu_i$ is
in~$\subcact{c_j}$, must use~$a_i$. Similarly, a \sa path
\emph{downwards and to the right} must use~$c_i$. Since for several
ancestor blocks of~$\mu_i$ the roots lie on both of these two kinds of
paths, we can bound the area containing them and show that it is
relatively small. This implies that the ancestor blocks are small as
well, providing a contradiction.

\begin{lemma}
  \label{lem:rab-geq90}
  Consider non-leaf blocks~$\mu_0,\mu_1,\mu_2$, such that~$\rt{\mu_1}=c_0$
  and~$\mu_2$ in~$\subcact{a_1}$; see Fig.~\ref{fig:rab-geq90}.
  \begin{myenum}
  \item[\emph{(i)}]~It is $\angle r_2 a_2 b_2, \angle r_2 c_2 b_2 \in
    [90\dg, 90\dg + \eps]$, $b_2$ lies to the right
    of~$\ray{r_2}{a_2}$ and to the left of~$\ray{r_2}{c_2}$.
  \item[\emph{(ii)}]~Each \sa $b_2 a_0$-path uses $a_2$; each \sa $b_2
    c_1$-path uses~$c_2$.
  \end{myenum}
\end{lemma}
\begin{proof}
  \sublem{1}~Assume $\angle r_2 a_2 b_2 < 90\dg$. Then, all \sa
    $b_2 a_0$ and~$b_2 c_1$-paths must use~$c_2$.  By
    Lemma~\ref{lem:angles}\sublem{4}, the lines through $a_0 c_2$ and
    $c_2 c_1$ are ``almost horizontal'', i.e., $\angle(\vec{a_0 c_2},
    \vright)$, $\angle(\vec{c_2 c_1}, \vright) \leq \eps/2$.  Since
    $r_2 c_2$ is ``almost vertical'', $r_2$ must lie below these lines
    and it is $\angle a_0 c_2 r_2$, $\angle c_1 c_2 r_2 \in
    [90\dg-\eps, 90\dg+\eps]$; see Fig.~\ref{fig:rab-geq90:2}. First,
    let~$b_2$ lie to the left of~$\ray{r_2}{c_2}$. Then, $b_2$ is
    above~$a_0 c_2$, and it is $\angle r_2 c_2 b_2 = \angle a_0 c_2
    r_2 + \angle a_0 c_2 b_2 \geq (90\dg - \eps) + 90\dg =180\dg -
    \eps$.  Now let $b_2$ lie to the right of~$\ray{r_2}{c_2}$. Then,
    $b_2$ is above $c_2 c_1$, and it is $\angle r_2 c_2 b_2 = \angle
    c_1 c_2 r_2 + \angle c_1 c_2 b_2 \geq (90\dg - \eps) + 90\dg =
    180\dg - \eps$.  Since $\eps < 22.5\dg$, this contradicts
    Lemma~\ref{lem:angles}\sublem{3}. Similarly, $\angle r_2 c_2 b_2
    \geq 90\dg$. Thus, by Lemma~\ref{lem:angles}\sublem{3}, $\angle
    r_2 a_2 b_2, \angle r_2 c_2 b_2 \in [90\dg, 90\dg+\eps]$.
    Since~$\angle a_2 b_2 c_2 \geq 90\dg$, $b_2$ lies to the right
    of~$\ray{r_2}{a_2}$ and to the left of~$\ray{r_2}{c_2}$.
    (If~$b_2$ lies to the left of both rays, it is~$\angle a_2b_2c_2 =
    \angle(\vec{a_2b_2},\vec{c_2b_2}) \leq 2\eps < 90\dg$.)
\smallskip

  \noindent\sublem{2}~Similarly, if a \sa $b_2 a_0$-path uses~$c_2$
    instead of~$a_2$, it is~$\angle r_2 c_2 b_2 \geq 180\dg - \eps$.
        The last part follows analogously. 
\end{proof}

The next lemma allows us to show that certain blocks are drawn smaller than their ancestors.
\begin{lemma}
  \label{lem:lengths}
  It holds:
  \begin{myenum}
  \item[\emph{\sublem{1}}]~$ \frac{|ra|}{|rc|}$, $\frac{|rc|}{|ra|}
    \geq \cos \eps$;
  \item[\emph{\sublem{2}}]~$ \frac{|ac|}{|ra|}, \frac{|ac|}{|rc|} \leq
    \tan \eps$;
  \item[\emph{\sublem{3}}]~The distance from~$a$ to the line
    through~$rc$ is at least $|ac| \cos \eps$.
  \item[\emph{\sublem{4}}]~Consider block~$\mu$ containing $a,b,c,r$,
    vertex~$u \neq a$ in~$\subcact{a}$ and~$v \neq c$
    in~$\subcact{c}$, $\deg(u) = \deg(v)=4$.  Then, $\frac{|au|}{|ac|}
    \leq \tan \eps$ and~$|uv| \leq (1 + 2\tan \eps)
    |ac|$. 
  \end{myenum}
\end{lemma}
\begin{proof}
\sublem{1}~Due to symmetry, we show only one part. By
  Lemma~\ref{lem:angles}\sublem{4}, $\angle a c r \in [90\dg - \eps,
  90\dg + \eps]$. Therefore,
  \begin{equation*}
\frac{|ra|}{|rc|} = \frac{\sin \angle acr}{\sin \angle
    rac} \geq \frac{\sin(90\dg - \eps)}{1} = \cos \eps.
\end{equation*}

\noindent\sublem{2}~It is~$\angle a r c \leq \eps$. Therefore, 
\begin{equation*}
\frac{|ac|}{|ra|} =
  \frac{\sin \angle arc}{\sin \angle acr} \leq \frac{\sin(\eps)}{\sin(90\dg -
    \eps)} = \tan \eps.
\end{equation*}

\noindent\sublem{3}~Let~$d$ be the point on the line through~$rc$
  minimizing~$|ad|$. Since~$\angle a c r \in [90\dg - \eps, 90\dg + \eps]$, it
  is~$\angle(\vec{ac}, \vec{ad}) \leq \eps$. Thus, $|ad| \geq |ac| \cos \eps$.
\smallskip

\noindent\sublem{4}~By Lemma~\ref{lem:angles}\sublem{4}, $\angle a c u
  \leq \eps$ and $\angle a u c \in [90\dg - \eps, 90\dg +
  \eps]$. Thus,
  \begin{equation*}
    \frac{|au|}{|ac|} = \frac{\sin \angle a c u}{\sin \angle a u c} \leq
    \frac{\sin \eps}{ \sin(90\dg - \eps)} = \tan \eps.
  \end{equation*}

 Similarly, $|vc| \leq |ac| \tan \eps$. Thus, it is $|uv| \leq
  |ua|+|ac|+|cv| \leq (1 + 2\tan \eps) |ac|$. 

\end{proof}

From now on, let~$\mu_0$ be the root block of~$\subcact{\bar{r}}$ and $\mu_1$,
$\mu_2$, $\mu_3$ its descendants such that~$\rt{\mu_1} = c_0$, $\rt{\mu_2} =
a_1$, $\rt{\mu_3} \in \{ a_2, c_2 \}$; see Fig.~\ref{fig:Gr}.  Light gray blocks
are the subject of Lemma~\ref{lem:cone-strip}\sublem{1}, which shows that
several ancestor roots lie inside a cone with a small angle. Dark gray blocks
are the subject of Lemma~\ref{lem:cone-strip}\sublem{2}, which considers the
intersection of the cones corresponding to a pair of sibling blocks and shows
that some of their ancestor roots lie inside a narrow strip; see
Fig.~\ref{fig:counterex:strip} for a sketch.

\newcommand{\LemConeStripText}{ Let~$\mu$ be a block in~$\subcact{c_2}$ with
  vertices~$a$, $b$, $c$, $\rt{\mu}$. \emph{\sublem{1}} Let~$\mu$ have depth~5
  in~$\subcact{\bar{r}}$. Then, the cone $\hp{b}{a} \cap \hp{b}{c}$
  contains~$\rt{\mu}$, $\rt{\pred{\mu}}$, $\rt{\gpred{2}{\mu}}$
  and~$\rt{\gpred{3}{\mu}}$. \emph{\sublem{2}}~Let $\mu$ have depth~4
  in~$\subcact{\bar{r}}$. There exist~$u$ in~$\subcact{a}$ and $v$ in
  $\subcact{c}$ of degree~4 and a strip~$S$ containing $\rt{\mu}$,
  $\rt{\pred{\mu}}$, $\rt{\gpred{2}{\mu}} = \rt{\mu_2}$, such that $u$ and~$v$
  lie on the different boundaries of~$S$. }
\begin{lemma}
  \label{lem:cone-strip}
    Let~$\mu$ be a block in~$\subcact{c_2}$ with vertices~$a$, $b$,
    $c$, $\rt{\mu}$.
  \begin{myenum}
  \item[\emph{\sublem{1}}]~Let~$\mu$ have depth~5
    in~$\subcact{\bar{r}}$. Then, the cone $\hp{b}{a} \cap \hp{b}{c}$
    contains~$\rt{\mu}$, $\rt{\pred{\mu}}$, $\rt{\gpred{2}{\mu}}$
    and~$\rt{\gpred{3}{\mu}}$.
  \item[\emph{\sublem{2}}]~Let $\mu$ have depth~4
    in~$\subcact{\bar{r}}$. There exist~$u$ in~$\subcact{a}$ and $v$
    in $\subcact{c}$ of degree~4 and a strip~$S$ containing
    $\rt{\mu}$, $\rt{\pred{\mu}}$, $\rt{\gpred{2}{\mu}} = \rt{\mu_2}$,
    such that $u$ and~$v$ lie on the different boundaries of~$S$.
  \end{myenum}
\end{lemma}
\begin{proof}
\noindent\sublem{1}~Consider a \sa $b b_0$-path~$\rho_0$ and a \sa $b
  b_1$-path~$\rho_1$.  By Lemma~\ref{lem:rab-geq90}\sublem{2}
  applied to~$\mu$, $b a$ is the first edge
  of~$\rho_0$ and $b c$ the first edge of~$\rho_1$. Since the
  cutvertices~$\rt{\mu}$, $\rt{\pred{\mu}}$, $\rt{\gpred{2}{\mu}}$,
  $\rt{\gpred{3}{\mu}}$ are on both~$\rho_0$ and~$\rho_1$, the
  statement holds. \smallskip

\noindent\sublem{2}~Consider blocks~$\mu_l$, $\mu_r$, such
  that~$\rt{\mu_l} = a$ and~$\rt{\mu_r} = c$. By~\sublem{1},
  $\rt{\mu}$, $\rt{\pred{\mu}}$, $\rt{\gpred{2}{\mu}}$ are in $\Lambda
  := \hp{b_l}{a_l} \cap \hp{b_l}{c_l} \cap \hp{b_r}{a_r} \cap
  \hp{b_r}{c_r}$. Let~$\vec{v_l}$ be the vector~$\vec{b_l c_l}$
  rotated by~$90\dg$ clockwise and $\vec{v_r}$ be the vector~$\vec{b_r
    a_r}$ rotated by~$90\dg$ counterclockwise. Note that by
  Lemma~\ref{lem:angles}\sublem{2},
  $\subcact{c_l}$, $\subcact{a_r}$ lie in~$\hp{b_l}{ c_l}$,
  $\hp{b_r}{a_r}$ respectively. Therefore, $\ray{c_l}{\vec{v_l}}$
  and~$\ray{a_r}{\vec{v_r}}$ (green resp. blue arrows in
  Fig.~\ref{fig:counterex:strip}) converge, since the converse would
  contradict Lemma~\ref{lem:diverge}. Let~$p$ be their intersection.
  Due to the chosen directions, $\rt{\mu}$, $\rt{\pred{\mu}}$,
  $\rt{\gpred{2}{\mu}}$ are below both~$c_l$ and~$a_r$.  Therefore,
  $\rt{\mu}$, $\rt{\pred{\mu}}$, $\rt{\gpred{2}{\mu}}$ are contained
  in the triangle $c_l a_r p$, which lies inside a strip~$S$ of width
  at most~$|c_l a_r|$, whose respective boundaries contain~$c_l$
  and~$a_r$. By Lemma~\ref{lem:lengths}\sublem{4} and~\sublem{2}, $
  |c_l a_r| \leq (1 + 2 \tan \eps) |a c| \leq (1 + 2 \tan \eps)(\tan
  \eps) \min \{ |\rt{\mu} a|, |\rt{\mu} c|\}$. 
   \end{proof}

\begin{figure}[tb]
  \hfill
  \subfloat[]{\includegraphics[page=1]{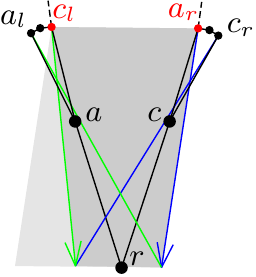} \label{fig:counterex:strip}}\hfill
  \subfloat[]{\includegraphics[page=2]{fig/counterex-lemmas.pdf} \label{fig:counterex:2strips}}\hfill
  \subfloat[]{\includegraphics[page=3]{fig/counterex-lemmas.pdf} \label{fig:counterex:diamond}}\hfill\null
  \caption{Showing the contradiction in Theorem~\ref{thm:no-sad}.}
\end{figure}

Again, we consider two siblings and the intersection of their corresponding
strips, which forms a small diamond containing the root of the ancestor block;
see Fig.~\ref{fig:counterex:2strips}, \ref{fig:counterex:diamond}.

\newcommand{\LemTwoStripsText}{ Consider block~$\mu = \mu_3$
  containing~$r=r(\mu),a,b,c$, and let $r_\pi := \rt{\pred{\mu_3}}$. It holds:
  \emph{\sublem{1}} $| r_\pi r| \leq \frac{(1 + 2 \tan \eps)(\tan \eps)^2}{\cos \eps}
  ( |r a|+ |r c|)$; \emph{\sublem{2}} $ |ra|,|rc| \leq |r r_\pi| (\tan \eps)^2
  $.}
\begin{lemma}
  \label{lem:2strips}
  \LemTwoStripsText
\end{lemma}
\begin{proof}
\noindent\sublem{1}~Let~$r_\pi = \rt{\pred{\mu}}$. Define~$d = (1 + 2
  \tan \eps)(\tan \eps)^2 |ac|$. Then, by
  Lemma~\ref{lem:lengths}\sublem{2} and~\sublem{4} and
  Lemma~\ref{lem:cone-strip}\sublem{2}, vertices~$a$, $r$ and~$r_\pi$ are
  contained in a strip~$s_1$ (green in
  Fig.~\ref{fig:counterex:2strips}) of width~$d$. Additionally, both
  boundaries of $s_1$ contain vertices of~$\subcact{a}$ (red dots),
  which lie in~$\hp{b}{a}$ and, by
  Lemma~\ref{lem:rab-geq90}\sublem{1}, to the left
  of~$\ray{r}{a}$. Thus, the downward direction
  along~$s_1$ is counterclockwise compared to~$\vec{a r}$. (Otherwise,
  the green strip could not contain~$a$.)
  Similarly, vertices~$c$, $r$ and~$r_0$ are contained in a
  strip~$s_2$ (blue) of width at most~$d$, and both boundaries of
  $s_2$ contain vertices of~$\subcact{c}$, which lie to the right
  of~$\ray{r}{c}$. Thus, the downward direction along~$s_2$ is
  clockwise compared to~$\vec{c r}$; see
  Fig.~\ref{fig:counterex:2strips}.

  Let us find an upper bound for the diameter of the
  parallelogram~$s_1 \cap s_2$. In the critical case, the right side
  of~$s_1$ touches~$ra$ and the left side of~$s_2$ touches~$rc$; see
  Fig.~\ref{fig:counterex:diamond}. Let~$a'$ (resp. $c'$) be the
  intersection of the right (resp. left) sides of~$s_1$ and~$s_2$,
  and~$r'$ the intersection of the left side of~$s_1$ and right side
  of~$s_2$. Let $d_a$ be the distance from~$a$ to the line through
  $rc$ and~$d_c $ the distance from~$c$ to the line through $ra$. By
  Lemma~\ref{lem:lengths}\sublem{3}, it is~$d_a, d_c \geq |ac| \cos
  \eps$. Moreover, it holds: $\frac{|r a'|}{|r a|} = \frac{d}{d_a}$
  and~$\frac{|r c'|}{|r c|} = \frac{d}{d_c}$. Therefore, $|ra'| \leq
  \frac{d |ra|}{|ac| \cos \eps}$, $|rc'| \leq \frac{d |rc|}{|ac| \cos
    \eps}$ and
  \begin{equation*}
    |r r'| \leq |r a'| + |r c'| \leq \frac{(1 + 2 \tan \eps)(\tan \eps)^2}{\cos
      \eps} (|ra| + |rc|).
  \end{equation*}
  Since $\angle a' r c' \leq \eps$, $r r'$ is the diameter, thus, $|r r_\pi|
  \leq |r r'|$. 
\smallskip

\noindent\sublem{2}~The claim follows from Lemma~\ref{lem:lengths}\sublem{2}
  and~\sublem{4}. 
\end{proof}

For~$\eps \leq 22.5\dg$, the two claims of Lemma~\ref{lem:2strips} contradict
each other. This concludes the proof of Theorem~\ref{thm:no-sad}.

\section{Planar Increasing-Chord Drawings of 3-Trees}
\label{sec:schnyder}

In this section, we show how to construct planar increasing-chord drawings of
3-trees. We make use of \emph{Schnyder labelings}~\cite{Schnyder1990} and
drawings of triangulations based on them. For a plane triangulation~$G=(V,E)$
with external vertices~$r,g,b$, its Schnyder labeling is an orientation and
partition of the interior edges into three trees~$T_r, T_g, T_b$ (called \emph{red},
\emph{green} and \emph{blue tree}), such that for each internal vertex $v$, its incident edges
appear in the following clockwise order: exactly one outgoing red, an arbitrary
number of incoming blue, exactly one outgoing green, an arbitrary number of
incoming red, exactly one outgoing blue, an arbitrary number of incoming
green. Each of the three outer vertices~$r,g,b$ serves as the root of the tree
in the same color and all its incident interior edges are incoming in the
respective color.
 For $v \in V$, let~$\regred{v}$ (the \emph{red region} of~$v$) denote
the region bounded by the $v g$-path in~$T_g$, the $v b$-path in~$T_b$ and
the edge~$gb$. Let~$|\regred{v}|$ denote
the number of the interior faces in~$\regred{v}$. The green and blue
regions~$\reggreen{v}$, $\regblue{v}$ are defined analogously. Assigning $v$ the
coordinates $(|\regred{v}|,|\reggreen{v}|,|\regblue{v}|) \in \R^3$ results in a
plane straight-line drawing of~$G$ in the plane $\{ x = (x_1,x_2,x_3) \mid x_1 +
x_2 + x_3 = f-1 \}$ called \emph{Schnyder drawing}. Here, $f$ denotes the number
of faces of~$G$. For a thorough introduction to this topic, see the book of
Felsner~\cite{Felsner2004}.

For~$\alpha,\beta \in [0\dg,360\dg]$, let~$[\alpha,\beta]$
denote the corresponding counterclockwise cone of directions. 
We consider drawings satisfying the following constraints.

\begin{definition}
  \label{def:alpha-Schnyder}
  Let~$G=(V,E)$ be a plane triangulated graph with a Schnyder
  labeling. For~$0\dg \leq \alpha \leq 60\dg$, we call an arbitrary
  planar straight-line drawing of $G$ \emph{$\alpha$-Schnyder}
  if for each internal vertex~$v \in V$, its outgoing red edge has
  direction in~$[90\dg-\frac{\alpha}{2},90\dg+\frac{\alpha}{2}]$, blue
  in~$[210\dg-\frac{\alpha}{2}, 210\dg+\frac{\alpha}{2}]$ and green
  in~$[330\dg-\frac{\alpha}{2},330\dg+\frac{\alpha}{2}]$ (see
  Fig.~\ref{fig:schnyder:cones}).
\end{definition}

According to Definition~\ref{def:alpha-Schnyder}, classical Schnyder
drawings are \aschnyder{$60\dg$}.  The next lemma shows an
interesting connection between $\alpha$-Schnyder and increasing-chord
drawings.

\begin{lemma}
  \label{lem:ic-schnyder}
  \aschnyder{$30\dg$} drawings are increasing-chord drawings. 
\end{lemma}
\begin{proof}
  Let~$G=(V,E)$ be a plane triangulation with a given Schnyder
  labeling and~$\Gamma$ a corresponding \aschnyder{$30\dg$}
  drawing. Let $r,g,b$ be the red, green
  and blue external vertex, respectively, and $T_r, T_g, T_b$  the directed trees
  of the corresponding color.

  Consider vertices~$s,t \in V$. First, note that monochromatic
  directed paths in~$\Gamma$ have increasing chords by
  Lemma~\ref{lem:path-ic90}.  Assume $s$ and $t$ are not connected by
  such a path. Then, they are both internal and $s$ is contained in
  one of the regions~$\regred{t}$, $\reggreen{t}$, $\regblue{t}$.
      Without loss of generality, we assume $s \in \regred{t}$.
      The $s r$-path in $T_r$ crosses the boundary of $\regred{t}$, and we
  assume without loss of generality that it crosses the blue boundary of
  $\regred{t}$ in $u \neq t$; see Fig.~\ref{fig:30schnyder:paths}.  The other
  cases are symmetric.

  Let~$\rho_r$ be the $s u$-path in~$T_r$ and $\rho_b$ the $t u$-path
  in~$T_b$; see Fig.~\ref{fig:30schnyder:2cones}. On the one hand, the direction
  of a line orthogonal to a segment of~$\rho_r$ is in~$[345\dg, 15\dg] \cup
  [165\dg,195\dg]$. On the other hand, $\rho_b$ is contained in a cone
  $[15\dg,45\dg]$ with apex~$u$. Thus, $\rho_b^{-1} \subseteq \front(\rho_r)$,
  and $\concat{\rho_r}{\rho_b^{-1}}$ is \sa by Fact~\ref{rem:contpath}.  By a
  symmetric argument it is also \sa in the other direction, and hence has
  increasing chords.
\end{proof}

\emph{Planar 3-trees} are the graphs
 obtained from a triangle by
repeatedly choosing a (triangular) face $f$, inserting a new vertex
$v$ into $f$, and connecting $v$ to each vertex of $f$. 

\begin{figure}[tb]
  \hfill
  \subfloat[]{\includegraphics[page=1]{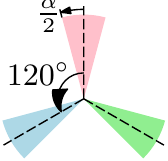} \label{fig:schnyder:cones}}
  \hfill
  \subfloat[]{\includegraphics[page=4]{fig/schnyder.pdf} \label{fig:30schnyder:paths}}
  \hfill
  \subfloat[]{\includegraphics[page=5]{fig/schnyder.pdf} \label{fig:30schnyder:2cones}}
  \hfill
  \subfloat[]{\includegraphics[page=2]{fig/schnyder.pdf} \label{fig:3trees:1}}
  \hfill
  \subfloat[]{\includegraphics[page=3]{fig/schnyder.pdf} \label{fig:3trees:2}}
  \hfill\null
  \caption{\protect\subref{fig:schnyder:cones}--\protect\subref{fig:30schnyder:2cones} $30\dg$-Schnyder drawings are increasing-chord;
 \protect\subref{fig:3trees:1},\protect\subref{fig:3trees:2} special case of planar 3-trees.}
\end{figure}

\begin{lemma}
  \label{lem:3trees}
  Planar 3-trees have \aschnyder{$\alpha$} drawings for any~$0\dg < \alpha \leq
  60\dg$.
\end{lemma}
\begin{proof}
  We describe a recursive construction of an \aschnyder{$\alpha$} drawing of a
  planar 3-tree. We start with an equilateral triangle and put a vertex~$v$ in
  its center. Then, we align the pattern from Fig.~\ref{fig:schnyder:cones}
  at~$v$. For the induction step, consider a triangular face~$xyz$ and assume
  that the pattern is centered at one of its vertices, say~$x$, such that the
  other two vertices are in the interiors of two distinct cones; see
  Fig.~\ref{fig:3trees:1}. It is now possible to move the pattern inside the
  triangle slightly, such that the same holds for all three vertices~$x,y,z$;
  see Fig.~\ref{fig:3trees:2}. We insert the new vertex at the center of the
  pattern and again get the situation as in Fig.~\ref{fig:3trees:1}. 
\end{proof}

Lemmas~\ref{lem:ic-schnyder} and~\ref{lem:3trees} provide a constructive proof for the following theorem.

\begin{theorem}
  \label{thm:3trees}
  Every planar 3-tree has a planar increasing-chord drawing.
\end{theorem}

\section{Self-Approaching Drawings in the Hyperbolic Plane}
\label{sec:hyperbolic}

Kleinberg~\cite{k-gruhs-07} showed that every tree can be drawn greedily
in the hyperbolic plane~$\Hp$. This is not the case in~$\R^2$. Thus,
$\Hp$ is more powerful than~$\R^2$ in this regard. Since \sa drawings
are closely related to greedy drawings, it is natural to investigate
the existence of \sa drawings in~$\Hp$.

We shall use the \emph{Poincar\'e disk} model for~$\Hp$, in
which~$\Hp$ is represented by the unit disk~$D = \{ x \in \R^2 : |x| <
1\}$ and lines are represented by circular arcs orthogonal to the
boundary of~$D$. For an introduction to the Poincar\'e disk model, see, for example, Kleinberg~\cite{k-gruhs-07} and the references therein.

First, let us consider a tree~$T=(V,E)$. A drawing of~$T$ in~$\R^2$ is
\sa if and only if no normal on an edge of~$T$ in any point crosses
another edge~\cite{acglp-sag-12}. The same condition holds in~$\Hp$.

\newcommand{\disth}[2]{d(#1,#2)}
\begin{lemma}
  \label{lem:H:normals}
  A straight-line drawing~$\Gamma$ of a tree~$T$ in~$\Hp$ is \sa if and only if
  no normal on an edge of~$T$ crosses $\Gamma$ in another point.
\end{lemma}
\begin{proof}
  The proof is similar to the Euclidean case. We present it for the sake of
  completeness. First, let~$\Gamma$ be a \sa drawing, for which the condition of
  the lemma is violated. \Wlog, let $\rho = (s,u, \dots, t)$ be the $s t$-path
  in~$T$, such that the normal on~$su$ in a point~$r$ crosses~$\rho$ in another
  point. Due to the piecewise linearity of~$\rho$, we may assume $r$ to be in
  the interior of~$su$. Let~$H_+ = \{ p=(p_x,p_y) \in D \mid p_y > 0 \}$ and
  $H_- = \{ p=(p_x,p_y) \in D \mid p_y < 0 \}$ the top and bottom hemispheres
  of~$D$. For~$p_1, p_2 \in D$, let~$\disth{p_1}{p_2}$ denote the hyperbolic
  distance between~$p_1$ and~$p_2$, i.e., the hyperbolic length of the
  corresponding geodesics. We recall the following basic fact whose proof is
  given, e.g., by Kleinberg~\cite{k-gruhs-07}.

  \begin{claim}
    Let~$0<y<1$, $p_- = (0,-y)$, $p_+ = (0,y)$. Then, for each~$p \in H_-$, it
    is $\disth{p}{p_-} < \disth{p}{p_+}$.
  \end{claim}

  Due to isometries, we can assume that~$r$ is in the origin of~$D$, $su$ is
  vertical, $s \in H_-$, $u \in H_+$. Let $a \in H_-$, $b \in H_+$ be two points
  on $su$, such that~$|ar|=|rb|$. Since the normal on~$su$ in~$r$
  crosses~$\rho$, there must exist a point~$c$ on~$\rho$, $c \in H_-$, such that
  $a,b,c$ are on~$\rho$ in this order. However, it is~$\disth{a}{c} <
  \disth{b}{c}$, a contradiction to~$\rho$ being \sa.

  Let~$\Gamma$ be a drawing of~$T$, for which the condition holds. Let~$a,b,c$
  be three consecutive points on a path~$\rho$ in~$\Gamma$. First, assume~$a,b$
  lie on the same arc of~$\Gamma$. We apply an isometry to~$\Gamma$, such that
  $ab$ is vertical, $a \in H_-$, $b \in H_+$, and $a,b$ are equidistant from the
  origin~$o$. The normal to~$\rho$ in~$o$ is the equator. Thus, it is~$c \notin
  H_-$, and~$\disth{b}{c} \leq \disth{a}{c}$. By applying this argument
  iteratively, this inequality also holds if~$a,b$ lie on different arcs. 
\end{proof}

 According to the
characterization by Alamdari et al.~\cite{acglp-sag-12}, some binary
trees have no self-approaching drawings in~$\R^2$.  We show that this
is no longer the case in~$\Hp$.

\begin{theorem}
  \label{thm:H:bintrees}
  Let $T=(V,E)$ be a tree, such that each node of~$T$ has degree either~1
  or~3. Then,~$T$ has a \sa drawing in~$\Hp$, in which every arc has the same
  hyperbolic length and every pair of incident arcs forms an angle of~$120\dg$.
\end{theorem}
\begin{proof}
  For convenience, we subdivide each edge of~$T$ once. We shall show that both
  pieces are collinear in the resulting drawing~$\Gamma$ and have the same
  hyperbolic length.

  \begin{figure}[tb]
    \hfill\subfloat[]{\includegraphics[scale=1.3, page=1]{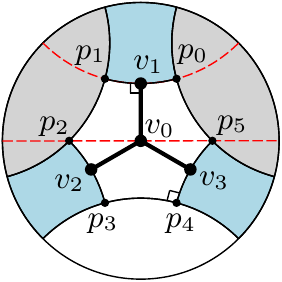}
      \label{fig:H:bintrees:1}}\hfill
    \hfill\subfloat[]{\includegraphics[scale=1.3, page=2]{fig/H-bintrees.pdf}
      \label{fig:H:bintrees:2}}\hfill
    \hfill\subfloat[]{\includegraphics[scale=1.3, page=3]{fig/H-bintrees.pdf}
      \label{fig:H:bintrees:3}}\hfill\null
    \caption{Constructing \ic drawings of binary trees and cactuses in $\Hp$.}
  \end{figure}

  First, consider a regular hexagon~$\hexagon = p_0 p_1 p_2 p_3 p_4 p_5$
  centered at the origin~$o$ of~$D$; see
  Fig.~\ref{fig:H:bintrees:1}. In~$\Hp$, it can have angles smaller
  than~$120\dg$. We choose them to be~$90\dg$ (any angle between $0\dg$ and
  $90\dg$ would work). Next, we draw a~$K_{1,3}$ with center~$v_0$ in~$o$ and
  the leaves~$v_1,v_2,v_3$ in the middle of the arcs~$p_0 p_1$, $p_2 p_3$, $p_4
  p_5$ respectively.

  For each such building block of the drawing consisting of a~$K_{1,3}$ inside a
  regular hexagon with $90\dg$ angles, we add its copy mirrored at an arc of the
  hexagon containing a leaf node of the tree constructed so far. For example, in
  the first iteration, we add three copies of~$\hexagon$ mirrored at~$p_0 p_1$, $p_2
  p_3$ and $p_4 p_5$, respectively, and the corresponding inscribed $K_{1,3}$
  subtrees. The construction after two iterations is shown in
  Fig.~\ref{fig:H:bintrees:2}. This process can be continued infinitely to
  construct a drawing~$\Gamma_\infty$ of the infinite binary tree. However, we
  stop after we have completed~$\Gamma$ for the tree~$T$.

  We now show that~$\Gamma_\infty$ (and thus also $\Gamma$) has the desired
  properties. Due to isometries and Lemma~\ref{lem:H:normals},
  it suffices to consider edge $e = v_0 v_1$ and show that a normal
  on~$e$ does not cross $\Gamma_\infty$ in another point. To see this, consider
  Fig.~\ref{fig:H:bintrees:1}. Due to the choice of the angles of~$\hexagon$,
  all the other hexagonal tiles of~$\Gamma_\infty$ are contained in one of the
  three blue quadrangular regions $\square_i := \hp{v_0}{v_i} \setminus
  (\hp{v_i}{p_{2i-1}} \cup \hp{v_i}{p_{2i-2}})$, $i=1,2,3$.  Thus, the
  regions~$\hp{v_1}{p_1}$ and~$\hp{v_1}{p_0}$ (gray) contain no point
  of~$\Gamma_\infty$.  Therefore, since each normal on~$v_0 v_1$ is contained in
  the ``slab'' $D \setminus (\hp{v_0}{v_1} \cup \hp{v_1}{v_0})$ bounded by the
  diameter through~$p_2,p_5$ and the line through~$p_0,p_1$ (dashed) and is
  parallel to both of these lines, it contains no other point
  of~$\Gamma_\infty$. 
\end{proof}

We note that our proof is similar in spirit to the one by
Kleinberg~\cite{k-gruhs-07}, who also used tilings of~$\Hp$ to prove that any
tree has a greedy drawing in~$\Hp$.

As in the Euclidean case, it can be easily shown that if a tree~$T$ contains a
node~$v$ of degree~4, it has a \sa drawing in~$\Hp$ if and only if~$T$ is a
subdivision of~$K_{1,4}$ (apply an isometry, such that~$v$ is in the origin
of~$D$). This completely characterizes the trees admitting a \sa drawing
in~$\Hp$. Further, it is known that every binary cactus and, therefore,
every 3-connected planar graph has a binary spanning
tree~\cite{angelini2010algorithm, Moitra2008}.

\begin{corollary}
  \emph{\sublem{1}}~A tree~$T$ has an \ic drawing in~$\Hp$ if and only if~$T$
    either has maximum degree~3 or is a subdivision of~$K_{1,4}$.
  \emph{\sublem{2}}~Every binary cactus and, therefore, every 3-connected planar
    graph has an \ic drawing in~$\Hp$.
\end{corollary}

Again, note that this is not the case for binary cactuses in~$\R^2$; see the
example in Theorem~\ref{thm:no-sad}. We use the above construction to produce
\emph{planar} \sa drawings of binary cactuses in~$\Hp$. We show how to choose a
spanning tree and angles at vertices of degree~2, such that non-tree edges can
be added without introducing crossings; see Fig.~\ref{fig:H:bintrees:3} for a
sketch.

\newcommand{\CorHypPlaneBinCactusText}{Every binary cactus has a planar \ic
  drawing in~$\Hp$.}
\begin{corollary}
 \CorHypPlaneBinCactusText
 \label{cor:HypPlaneBinCactus}
\end{corollary}

\begin{proof}
  \Wlog, let~$G$ be a binary cactus rooted at block~$\nu$ such that each
  block~$\mu$ of~$G$ is either a single edge or a cycle.  For each block~$\mu$ forming
  a cycle~$v_0$, $v_1$, \dots, $v_k, v_0$, $\rt{\mu} = v_0$, we remove edge~$v_0
  v_k$, thus obtaining a binary tree~$T$.  We embed it similar to the proof of
  Theorem~\ref{thm:H:bintrees} such that additionally the counterclockwise angle $\angle v_{j-1} v_j
  v_{j+1} = 120\dg$ for $j=1, \dots, k-1$. Obviously,~$T$ is
  drawn in a planar way since for each edge~$e$ of~$T$, each half of~$e$ is drawn inside
  its hexagon.  

  It remains to show that for each~$\mu$, adding arc~$v_0 v_k$ introduces no
  crossings.  For each~$j = 1, \dots, k-1$, we can apply an isometry to the
  drawing, such that $v_j$ is in the origin and~$\vec{v_j v_{j+1}}$ points
  upwards; see Fig.~\ref{fig:H:bintrees:3}. According to the construction
  of~$T$, subcactus~$\subcactb{v_0}{\mu}$ (maximal subcactus of~$G$
  containing~$v_0$ and no other vertex of~$\mu$) lies in the green region
  contained in~$\hp{v_1}{v_0}$ and~$\subcactb{v_k}{\mu}$ in the blue region
  contained in~$\hp{v_{k-1}}{v_k}$. Since $v_0 \notin \hp{v_{k-1}}{v_k}$
  and~$v_k \notin \hp{v_1}{v_0}$, arc~$v_0 v_k$ crosses
  neither~$\subcactb{v_0}{\mu}$ nor~$\subcactb{v_k}{\mu}$. Furthermore, $v_0$
  and~$v_k$ lie inside the~$120\dg$ cone~$\Lambda_j$ formed
  by~$\ray{v_j}{v_{j+1}}$ and~$\ray{v_{j}}{v_{j-1}}$.
 Thus, $v_0 v_k$ does not
  cross~$v_{j-1}{v_j}$, $v_jv_{j+1}$. Since subcactus~$\subcactb{v_j}{\mu}$ is
  in $\Hp \setminus \Lambda_j$ (it lies in the red area in
  Fig.~\ref{fig:H:bintrees:3}), it is not crossed by~$v_0 v_k$
  either. 
\end{proof}

\section{Conclusion}

We have studied the problem of constructing \sa and \ic drawings of
3-connected planar graphs and triangulations in the Euclidean and
hyperbolic plane.  Due to the fact that every such graph has a
spanning binary cactus, and in the case of a triangulation even one
that has a special type of triangulation (downward-triangulation), \sa
and \ic drawings of binary cactuses played an important role.

We showed that, in the Euclidean plane, downward-triangulated binary
cactuses admit planar \ic drawings, and that the condition of being
downward-triangulated is essential as there exist binary cactuses that
do not admit a (not necessarily planar) \sa drawing.  Naturally, these
results imply the existence of non-planar \ic drawings of
triangulations.  It remains open whether every 3-connected planar
graph has a \sa or \ic drawing. If this is the case, according to our
example in Theorem~\ref{thm:no-sad}, the construction must be
significantly different from both known
proofs~\cite{angelini2010algorithm, Moitra2008} of the weak
Papadimitriou-Ratajczak conjecture~\cite{Papadimitriou2005} (you
cannot just take an arbitrary spanning binary cactus) and would prove
a stronger statement.

For planar 3-trees, which are special triangulations, we introduced
$\alpha$-Schnyder drawings, which have increasing chords for $\alpha \le 30\dg$,
to show the existence of planar \ic drawings.  It is an open question whether
this method works for further classes of triangulations.  Which triangulations admit
$\alpha$-Schnyder drawings for arbitrarily small values of $\alpha$ or for $\alpha = 30\dg$?

Finally, we studied drawings in the hyperbolic plane.  Here we gave a
complete characterization of the trees that admit an \ic drawing (which
then is planar) and used it to show the existence of non-planar \ic
drawings of 3-connected planar graphs.  For binary cactuses even a
planar \ic drawing exists.

It is worth noting that all \sa drawings we constructed are actually
\ic drawings.  Is there a class of graphs that admits a \sa drawing
but no \ic drawing?


{\small


\begin{thebibliography}{10}

\bibitem{acglp-sag-12}
S.~Alamdari, T.~M. Chan, E.~Grant, A.~Lubiw, and V.~Pathak.
\newblock Self-approaching graphs.
\newblock In W.~Didimo and M.~Patrignani, editors, {\em Graph Drawing (GD'12)},
  volume 7704 of {\em LNCS}, pages 260--271. Springer, 2013.

\bibitem{acbfp-mdg-2012}
P.~{Angelini}, E.~{Colasante}, G.~{Di Battista}, F.~{Frati}, and
  M.~{Patrignani}.
\newblock Monotone drawings of graphs.
\newblock {\em J. Graph Algorithms Appl.}, 16(1):5--35, 2012.

\bibitem{Angelini2012}
P.~Angelini, G.~Di~Battista, and F.~Frati.
\newblock Succinct greedy drawings do not always exist.
\newblock {\em Networks}, 59(3):267--274, 2012.

\bibitem{adkmrsw-mdgfe-2012}
P.~Angelini, W.~Didimo, S.~Kobourov, T.~Mchedlidze, V.~Roselli, A.~Symvonis,
  and S.~\mbox{Wismath}.
\newblock Monotone drawings of graphs with fixed embedding.
\newblock In M.~van Kreveld and B.~\mbox{Speckmann}, editors, {\em Graph Drawing
  (GD'11)}, volume 7034 of {\em LNCS}, pages 379--390. Springer, 2012.

\bibitem{angelini2010algorithm}
P.~Angelini, F.~Frati, and L.~Grilli.
\newblock An algorithm to construct greedy drawings of triangulations.
\newblock {\em J. Graph Algorithms Appl.}, 14(1):19--51, 2010.

\bibitem{dfg-icgps-14}
H.~R. Dehkordi, F.~Frati, and J.~Gudmundsson.
\newblock Increasing-chord graphs on point sets.
\newblock In C.~Duncan and A.~Symvonis, editors, {\em Graph Drawing (GD'14)},
  volume 8871 of {\em LNCS}. Springer, 2014.
\newblock To appear.

\bibitem{Dhandapani2010}
R.~Dhandapani.
\newblock Greedy drawings of triangulations.
\newblock {\em Discrete Comput. Geom.}, 43:375--392, 2010.

\bibitem{Eppstein2011}
D.~Eppstein and M.~T. Goodrich.
\newblock Succinct greedy geometric routing using hyperbolic geometry.
\newblock {\em IEEE Trans. Computers}, 60(11):1571--1580, 2011.

\bibitem{Felsner2004}
S.~Felsner.
\newblock {\em Geometric Graphs and Arrangements}, chapter~2, pages 17--42.
\newblock Advanced Lectures in Mathematics. Vieweg+Teubner Verlag, 2004.

\bibitem{Goodrich2009}
M.~T. Goodrich and D.~Strash.
\newblock Succinct greedy geometric routing in the {E}uclidean plane.
\newblock In Y.~Dong, D.-Z. Du, and O.~Ibarra, editors, {\em Algorithms and
  Computation (ISAAC'09)}, volume 5878 of {\em LNCS}, pages 781--791. Springer,
  2009.

\bibitem{heh-grbgt-09}
W.~Huang, P.~Eades, and S.-H. Hong.
\newblock A graph reading behavior: Geodesic-path tendency.
\newblock In {\em {IEEE} Pacific Visualization Symposium (PacificVis'09)},
  pages 137--144, 2009.

\bibitem{ikl-sac-99}
C.~Icking, R.~Klein, and E.~Langetepe.
\newblock Self-approaching curves.
\newblock {\em Math. Proc. Camb. Phil. Soc.}, 125:441--453, 1999.

\bibitem{kssw-mdt-14}
P.~Kindermann, A.~Schulz, J.~Spoerhase, and A.~Wolff.
\newblock On monotone drawings of trees.
\newblock In C.~Duncan and A.~Symvonis, editors, {\em Graph Drawing (GD'14)},
  volume 8871 of {\em LNCS}. Springer, 2014.
\newblock To appear.

\bibitem{k-gruhs-07}
R.~Kleinberg.
\newblock Geographic routing using hyperbolic space.
\newblock In {\em Computer Communications (INFOCOM'07)}, pages 1902--1909,
  2007.

\bibitem{lppfh-ttgv-06}
B.~Lee, C.~Plaisant, C.~S. Parr, J.-D. Fekete, and N.~Henry.
\newblock Task taxonomy for graph visualization.
\newblock In {\em AVI Workshop on Beyond Time and Errors: Novel Evaluation
  Methods for Information Visualization (BELIV'06)}, pages 1--5. ACM, 2006.

\bibitem{Moitra2008}
A.~Moitra and T.~Leighton.
\newblock Some results on greedy embeddings in metric spaces.
\newblock In {\em Foundations of Computer Science (FOCS’08)}, pages 337--346.
  IEEE Computer Society, 2008.

\bibitem{np-egdt-2013}
M.~N\"ollenburg and R.~Prutkin.
\newblock Euclidean greedy drawings of trees.
\newblock In H.~Bodlaender and G.~Italiano, editors, {\em Algorithms (ESA'13)},
  volume 8125 of {\em LNCS}, pages 767--778. Springer, 2013.

\bibitem{Papadimitriou2005}
C.~H. Papadimitriou and D.~Ratajczak.
\newblock On a conjecture related to geometric routing.
\newblock {\em Theor. Comput. Sci.}, 344(1):3--14, 2005.

\bibitem{phnk-ulgd-12}
H.~C. Purchase, J.~Hamer, M.~N{\"o}llenburg, and S.~G. Kobourov.
\newblock On the usability of {L}ombardi graph drawings.
\newblock In W.~Didimo and M.~Patrignani, editors, {\em Graph Drawing (GD'12)},
  volume 7704 of {\em LNCS}, pages 451--462. Springer, 2013.

\bibitem{Rao2003}
A.~Rao, S.~Ratnasamy, C.~Papadimitriou, S.~Shenker, and I.~Stoica.
\newblock Geographic routing without location information.
\newblock In {\em Mobile Computing and Networking (MobiCom'03)}, pages 96--108.
  ACM, 2003.

\bibitem{r-cic-94}
G.~Rote.
\newblock Curves with increasing chords.
\newblock {\em Math. Proc. Camb. Phil. Soc.}, 115:1--12, 1994.

\bibitem{Schnyder1990}
W.~Schnyder.
\newblock Embedding planar graphs on the grid.
\newblock In {\em Discrete Algorithms (SODA'90)}, pages 138--148. SIAM, 1990.

\bibitem{Wang2014}
J.-J. Wang and X.~He.
\newblock Succinct strictly convex greedy drawing of 3-connected plane graphs.
\newblock {\em Theor. Comput. Sci.}, 532:80--90, 2014.

\end{thebibliography}

}
\end{document}